\pdfoutput=1

\PassOptionsToPackage{colorlinks,linkcolor={blue},citecolor={blue},urlcolor={red},breaklinks=true,final}{hyperref}

\documentclass[a4paper,UKenglish,cleveref,thm-restate,nosumlimits,nonamelimits,final]{lipics-v2021}

\hideLIPIcs

\theoremstyle{definition}
\newtheorem{defn}[theorem]{Definition}

\usepackage{chngcntr}
\usepackage{apptools}
\AtAppendix{\counterwithin{lemma}{section}}
\AtAppendix{\counterwithin{theorem}{section}}
\AtAppendix{\counterwithin{proposition}{section}}

\title{Quantitative Hennessy-Milner Theorems via Notions of Density}

\author{Jonas Forster, Sergey Goncharov}{Friedrich-Alexander-Universität Erlangen-Nürnberg,
  Germany}{}{}{}

\author
{Dirk Hofmann}{Center for Research and Development in Mathematics and Applications, University of Aveiro, Portugal}{dirk@ua.pt}{}{}

\author{Pedro Nora, Lutz Schr{\"o}der, Paul Wild}{Friedrich-Alexander-Universität Erlangen-Nürnberg,
  Germany}{}{}{}

\authorrunning{J. Forster, S. Goncharov, D. Hofmann, P. Nora, L. Schr{\"o}der, P. Wild}



\nolinenumbers 

\usepackage{stmaryrd}
\usepackage{mathtools}
\usepackage{tikz}
\usepackage{tikz-cd}

\usepackage[square,numbers]{natbib}

\usepackage{mathrsfs}

\usetikzlibrary{automata, positioning, arrows}
\newcommand*\dif{\mathop{}\!\mathrm{d}}

\newcommand{\sem}[1]{\llbracket #1 \rrbracket}

\newcommand{\weighted}{\mathcal{W}}
\newcommand{\kantorovich}{\mathcal{K}}
\newcommand{\paraP}{\mathcal{B}}
\newcommand{\pow}{\mathcal{P}}

\usepackage{hyperref}
\hypersetup{pdftex,colorlinks=true,allcolors=blue}
\usepackage{hypcap}
\usepackage{array}
\usepackage{microtype}

\def\gettexliveversion#1(#2 #3 #4#5#6#7#8)#9\relax{#4#5#6#7}
\edef\texliveversion{\expandafter\gettexliveversion\pdftexbanner\relax}

\ifnum\texliveversion<2018
\RequirePackage{chngcntr}
\fi

\RequirePackage[shortlabels]{enumitem}
\counterwithin*{equation}{section}

\newlist{tfae}{enumerate}{1}
\setlist[tfae,1]{label=(\roman*)}
\def\nlabel#1#2{\begingroup #2
  \def\@currentlabel{#2}
  \phantomsection\label{#1}\endgroup
}

\newlist{conditions}{description}{1}
\setlist[conditions]{font=\normalfont\space,labelindent=\parindent}

\newcommand{\calI}{\mathcal{I}}
\newcommand{\calJ}{\mathcal{J}}
\newcommand{\calL}{\mathcal{L}}

\newcommand{\cId}{\mathbf{Id}}
\newcommand{\cL}{\mathbf{L}}
\newcommand{\cFun}{\mathbf{Fun}}
\newcommand{\cC}{\mathbf{C}}
\newcommand{\cInf}{\mathbf{Inf}}
\newcommand{\ccInf}{\mathbf{cInf^\vee}}

\RequirePackage{mathtools}

\DeclareSymbolFont{symbolsA}{U}{txsya}{m}{n}
\DeclareSymbolFont{symbolsC}{U}{txsyc}{m}{n}
\DeclareMathSymbol{\multimapdot}{\mathrel}{symbolsC}{20}
\DeclareMathSymbol{\multimapdotinv}{\mathrel}{symbolsC}{21}
\DeclareMathSymbol{\multimap}{\mathrel}{symbolsA}{40}
\DeclareMathSymbol{\multimapinv}{\mathrel}{symbolsC}{18}

\newcommand{\SET}{\catfont{Set}}

\newcommand{\ORD}{\catfont{Ord}}

\newcommand{\EQ}{\catfont{Equ}}

\newcommand{\POST}{\catfont{Pos}}
\newcommand{\MET}{\catfont{Met}}
\newcommand{\UMET}{\catfont{UMet}}
\newcommand{\BMET}{\catfont{BMet}}

\newcommand{\Cats}[1]{#1\text{-}\catfont{Cat}}

\newcommand{\Coalg}[1]{\catfont{CoAlg}(#1)}

\newcommand{\ar}{\mathrm{ar}}

\newcommand{\sep}{\mathsf{sep}}

\newcommand{\sym}{\mathsf{sym}}

\newcommand{\two}{\catfont{2}}

\newcommand{\V}{\mathcal{V}}

\newcommand{\catfont}[1]{\mathsf{#1}}

\makeatletter

\def\slashedarrowfill@#1#2#3#4#5{
  $\m@th\thickmuskip0mu\medmuskip\thickmuskip\thinmuskip\thickmuskip
   \relax#5#1\mkern-7mu
   \cleaders\hbox{$#5\mkern-2mu#2\mkern-2mu$}\hfill
   \mathclap{#3}\mathclap{#2}
   \cleaders\hbox{$#5\mkern-2mu#2\mkern-2mu$}\hfill
   \mkern-7mu#4$
}

\newcommand*{\rightrelarrowfill@}{\slashedarrowfill@\relbar\relbar{\raisebox{0pc}{$\mapstochar$}}\rightarrow}
\newcommand*{\xrelto}[2][]{\ext@arrow 0055{\rightrelarrowfill@}{\;#1\;}{\;#2\;}}

\newcommand*{\rightmodarrowfill@}{\slashedarrowfill@\relbar\relbar{\raisebox{0pc}{$\hspace{1pt}\circ$}}\rightarrow}
\newcommand*{\xmodto}[2][]{\ext@arrow 0055{\rightmodarrowfill@}{\;#1\;}{\;#2\;}}

\newcommand*{\rightkrelarrowfill@}{\slashedarrowfill@\relbar\relbar{\raisebox{0pc}{$\hspace{1pt}\mapstochar$}}\rightharpoonup}
\newcommand*{\xkrelto}[2][]{\ext@arrow 0055{\rightkrelarrowfill@}{\;#1\;}{\;#2\;}}

\newcommand*{\rightkmodarrowfill@}{\slashedarrowfill@\relbar\relbar{\raisebox{0pc}{$\hspace{1pt}\circ$}}\rightharpoonup}
\newcommand*{\xkmodto}[2][]{\ext@arrow 0055{\rightkmodarrowfill@}{\;#1\;}{\;#2\;}}

\makeatother

\newcommand{\ftF}{\functorfont{F}}
\newcommand{\ftG}{\functorfont{G}}

\DeclareMathOperator{\bd}{bd}

\newcommand{\ftII}[1]{{\catfont{|}{#1}\catfont{|}}}

\newcommand{\ftbF}{\functorfont{\overline{F}}}

\newcommand{\functorfont}{\mathsf}

\RequirePackage{bbm}

\DeclareMathAlphabet{\mathpzc}{OT1}{pzc}{m}{it}

\RequirePackage{dsfont}

\newcommand{\df}[1]{\emph{\textbf{#1}}}

\tikzset{
  relational/.style={
    outer sep=3pt,
    decoration={
      markings,
      mark=at position 0.5 with {\node[transform shape] (tempnode) {\tiny $\rvert$};},
    },
    postaction={decorate},
  },
}

\tikzset{
  distrib/.style={
    outer sep=3pt,
    decoration={
      markings,
      mark=at position 0.5 with {\node[transform shape] (tempnode) {\tiny
          \rmfamily o};},
    },
    postaction={decorate},
  },
}

\tikzset{
  symbol/.style={
    draw=none, every to/.append style={
      edge node={node [sloped, allow upside down,auto=false]{$#1$}}} } }

\newcommand{\Sema}[1]{\llbracket{#1}\rrbracket_\alpha}
\newcommand{\Semb}[1]{\llbracket{#1}\rrbracket_\beta}

\begin{document}

\maketitle

\begin{abstract}
  The classical \emph{Hennessy-Milner theorem} is an important tool in
  the analysis of concurrent processes; it guarantees that any two
  non-bisimilar states in finitely branching labelled transition
  systems can be distinguished by a modal formula.  Numerous variants
  of this theorem have since been established for a wide range of logics and
  system types, including quantitative versions where lower bounds on
  behavioural distance (e.g.~in weighted, metric, or probabilistic
  transition systems) are witnessed by quantitative modal
  formulas. Both the qualitative and the quantitative versions have
  been accommodated within the framework of \emph{coalgebraic logic},
  with distances taking values in quantales, subject to certain
  restrictions, such as being so-called \emph{value quantales}. While
  previous quantitative coalgebraic Hennessy-Milner theorems apply
  only to liftings of set functors to (pseudo-)metric spaces, in the
  present work we provide a quantitative coalgebraic Hennessy-Milner
  theorem that applies more widely to functors native to metric spaces;
  notably, we thus cover, for the first time, the well-known
  Hennessy-Milner theorem for continuous probabilistic transition
  systems, where transitions are given by Borel measures on metric
  spaces, as an instance. In the process, we also relax the
  restrictions imposed on the quantale, and additionally parametrize
  the technical account over notions of \emph{closure} and, hence,
  \emph{density}, providing associated variants of the Stone-Weierstraß theorem; this allows us to cover, for instance, behavioural
  ultrametrics.

\end{abstract}

\section{Introduction}

Modal logic in general is an established tool in the analysis of
concurrent systems. One of its uses is as a means to distinguish
non-equivalent states; for instance, the classical Hennessy-Milner
theorem~\cite{HennessyMilner85} guarantees that any two non-bisimilar
states in finitely branching labelled transition systems can be
distinguished by a formula in a modal logic naturally associated to
labelled transition systems. Similar theorems have subsequently
proliferated, having been established, for instance, for probabilistic
transition systems~\cite{DBLP:journals/iandc/LarsenS91}, neighbourhood
structures~\cite{HansenEA09}, and open bisimilarity in the
$\pi$-calculus~\cite{AhnEA17}. As a recent example application, the
counterproof for unlinkability in the ICAO 9303 standard for
e-passports~\cite{FilimonovEA19} is based on providing a
distinguishing modal formula in an intuitionistic modal logic.

For systems featuring quantitative data, such as probabilistic or
weighted systems or metric transition systems, \emph{behavioural
  distance} provides a more fine-grained measure of agreement between
systems than two-valued bisimilarity
(e.g.~\cite{GiacaloneEA90,DBLP:journals/tcs/BreugelW05,DBLP:journals/tcs/BreugelHMW07,AlfaroEA09}). In
analogy to the classical Hennessy-Milner theorem, behavioural
distances can often be characterized by quantitative modal logics, in
the sense that the behavioural distance of any two states can be
approximated by the difference in value of quantitative modal formulae
on these states (that is, for states with distance $>r$ one can find a
quantitative modal formula on which the states disagree by at
least~$r$). Such theorems, to which we refer as \emph{quantitative
  Hennessy-Milner theorems}, have been proved, e.g., for probabilistic
transition
systems~\cite{DBLP:journals/tcs/BreugelW05,DBLP:journals/tcs/BreugelHMW07}
and for metric transition systems~\cite{AlfaroEA09}.

\emph{Universal coalgebra}~\cite{Rutten00} serves as a generic
framework for concurrent systems, based on the key abstraction of
encapsulating the system type in a functor, whose coalgebras then
correspond to the systems of interest. Both two-valued and
quantitative Hennessy-Milner theorems have been established at the
level of generality offered by \emph{coalgebraic modal logic}. The
two-valued coalgebraic Hennessy-Milner theorem~\cite{Pat04,Sch08}
covers all coalgebraic system types, under the assumption of having a
\emph{separating} set of modalities; instances include the mentioned
Hennessy-Milner theorems for probabilistic systems~\cite{DBLP:journals/iandc/LarsenS91}
and neighbourhood structures~\cite{HansenEA09}. Various quantitative
coalgebraic Hennessy-Milner theorems have been established fairly
recently~\cite{KonigMikaMichalski18,WS20,WS21,KKK+21}. These existing
theorems are tied to considering liftings of set functors to metric
spaces (or in fact more general topological categories~\cite{KKK+21});
our contribution in the present work is to complement these theorems
by a result that instead applies to unrestricted functors on metric
spaces (we give a more detailed comparison in the related work
section). In particular, our result covers, for the first time, the
original expressiveness result for probabilistic modal logic on
continuous probabilistic transition systems (where `continuous' refers
to the structure of the state
space)~\cite{DBLP:journals/tcs/BreugelW05,DBLP:journals/tcs/BreugelHMW07}
as an instance. We work not only in coalgebraic generality but also
parametrize the development over the choice of a \emph{quantale}~$\V$,
in which distances and truth values are taken; this covers the case of
standard bounded real-valued distances by taking~$\V$ to be the unit
interval, and the classical two-valued case by taking~$\V$ to be the
set of Boolean truth values. Previous work on quantalic
distances~\cite{WS21} needed to restrict to so-called value
quantales~\cite{Flagg97}; we relax this assumption, covering, for
instance, all finite quantales (such as the four-valued quantale used
in some paraconsistent logics; see,
e.g.,~\cite{RivieccioEA17}), and the square of the unit interval.

Technically, our results are additionally parametrized over the choice
of \emph{closure} operators on sets of $\V$-valued predicates, which
induce a notion of \emph{density}. The notion of density is the key
ingredient that lets our results apply beyond discrete state spaces
(e.g.\ to the mentioned continuous probabilistic transition systems);
by varying the notion of closure, we cover, for instance, both
standard metric spaces and ultrametric spaces (which in turn are
induced by different quantale structures on the unit interval).

Proofs are mostly omitted from the main body of the paper but can be
found in the appendix.

\subparagraph*{Related work}

As indicated above, quantatitive coalgebraic Hennessy-Milner theorems
exist in previous work~\cite{KonigMikaMichalski18,WS20,WS21,KKK+21},
from which our present work is distinguished in that it applies to
functors that live natively on metric spaces (such as tight Borel
distributions) rather than only to liftings of set functors (such as
finitely supported distributions). We detach the technical development
from both lax extensions~\cite{WS20,WS21} and fixpoint
induction~\cite{KonigMikaMichalski18,WS20,WS21}, which work only for
monotone modalities; we thereby cover also systems requiring
non-monotone modalities, such as weighted transition systems with
negative weights. A recent general framework for Hennessy-Milner
theorems based on Galois connections between real-valued predicates
and (pseudo-)metrics is aimed primarily at generality over a
linear-time/branching-time spectrum~\cite{BeoharEA22}.

The framework of codensity liftings developed by Komorida et
al.~\cite{KKK+21} works at a very high level of generality, and in
fact applies to topological categories (or
$\mathsf{CLat}_{\sqcap}$-fibrations, in the terminology of
\emph{op.~cit.}) beyond metric spaces, such as uniform spaces. Our
present framework is on the one hand more general in that we do not
restrict to functors lifted from the category of sets (\emph{fibred}
functors, in the terminology of fibrations), but on the other hand
less general in that we cover only (quantalic) behavioural
distances. In terms of the main technical result, we provide a
coalgebraic quantitative Hennessy-Milner theorem that is stated in
fairly simple terms, and can be instantiated to concrete logics and
systems by just verifying a few fairly straightforward conditions that
concern only the functor and the modalities. In particular, we have no
conditions requiring that certain sets of formula evaluations on a
given coalgebra are \emph{approximating} (cf.~\cite[Theorems~IV.5,
IV.7]{KKK+21}); instead, we prove similar properties as \emph{lemmas}
along the way, using the key notions of closure and density. In fact,
one of the conditions of our Hennessy-Milner theorem can be seen as a
form of Stone-Weierstraß property, and in particular concerns density of
sets of functions closed under suitable propositional combinations;
this property depends only on the quantale, not on the functor or the
modalities, and we give general Stone-Weierstraß-type theorems for
several classes of quantales.

\section{Preliminaries}

Basic familiarity with category theory will be assumed~\cite{AHS90,Awodey10}.
More specifically, we make an extensive use of topological categories and some
background results concerning them (see e.g.~\citep{AHS90}).
Let us recall some preliminaries, relevant for our coalgebraic modelling.

For an endofunctor \(\ftF\) on a category \(\mathcal{C}\), an
\df{\(\ftF\)-coalgebra} consists of an object~\(X\) of
\(\mathcal{C}\), thought of as an object of
\emph{states},
and a morphism $\alpha \colon X \rightarrow \ftF X$, thought of as
assigning structured collections (set, distributions, etc.)  of
successors to states. A \df{coalgebra morphism} from \((X, \alpha)\)
to \((Y, \beta)\) is a morphism $f \in \mathcal{C}(X, Y)$ such that
\(\beta \circ f = \ftF f \circ \alpha\).

A \df{concrete category} over \(\SET\) comes equipped with a faithful functor $\ftII{-}\colon \mathcal{C} \rightarrow \SET$,
which allows us to speak about individual \emph{states}, as elements of $\ftII{X}$.

Given a coalgebra \((X, \alpha)\) and states \(x,y \in \ftII{X}\), we say that \(x\) and \(y\) are \df{behaviourally equivalent}
if there are a coalgebra \((Z, \gamma)\) and a coalgebra morphism \(f\colon X \rightarrow Z\) such that \(\ftII{f}(x) = \ftII{f}(y)\). (For brevity, we restrict the treatment of both behavioural equivalence and behavioural distances to states in the same coalgebra; in all our examples the extension to states in different coalgebras can be accommodated by taking coproducts.)
	The notion of behavioural equivalence is strictly two-valued, meaning that different states are either behaviourally equivalent or not.
	The downside of this notion is thus that in systems dealing with quantitative information, any slight change can render two states behaviourally distinct, even though they may be virtually indistinguishable in any practical context. The rest of this paper is concerned with quantifying the degree to which states differ from each other, as well as with
 logics to witness these degrees.
\begin{example}\label{exa:main}
	\begin{enumerate}
		\item\label{main:lts} Labelled transition systems w.r.t.\ a set of actions \(A\) are coalgebras for the \(\SET\)-functor \(\mathcal{P}(A \times {-})\). Behavioural equivalence coincides with the classical notion of (strong) bisimilarity.

		\item\label{main:para} Let $\paraP$ be the \emph{$\rotatebox{45}{\scalebox{.75}{$\Box$}}$-valued powerset functor}, i.e.\ the functor that sends a set~$X$ to the set of maps $X\to \{\bot,\mathsf{N},\mathsf{B},\top\}$. Every such map $q$ can be seen as a
		sort of a fuzzy set, which for every element $x$ tells either that $x$ is in the set ($q(x)=\top$),
		or it is not in the set ($q(x)=\bot$), or there is evidence that $x$ is in the set
		and that it is not ($q(x)=\mathsf{B}$), or nothing is known ($q(x)=\mathsf{N}$). This is
		indeed a functor, due to the lattice structure on $\{\bot,\mathsf{N},\mathsf{B},\top\}$,
		which is diamond-like: $\bot$ and $\top$ are the bottom and the top and $\mathsf{N},\mathsf{B}$ are
		the left and the right corners. $\paraP$-coalgebras have been used for
		defining a Kripke-style semantics of paraconsistent modal logics~\cite{RivieccioEA17}.

		\item\label{main:prob} We denote by \(\mathcal{D}\) the functor that maps a set X to the set of finitely supported probability distributions on $X$. Coalgebras for the functor \(FX = (1 + \mathcal{D}X)^A\), for a finite set of actions \(A\), are probabilistic transition systems~\cite{DBLP:journals/iandc/LarsenS91,DBLP:journals/iandc/DesharnaisEP02}. In this context, behavioural equivalence instantiates to probabilistic bisimilarity.

		\item\label{main:weight} We consider weighted transition systems with possibly negative weights (e.g.~\cite{BouyerEA08}):
Let \(\weighted\)  be the functor on 1-bounded metric spaces that maps every set \(X\) to the set of finite \([-1,1]\)-weighted sets over \(X\). That is, the elements of \(\weighted X\) are functions \(t\colon X \to [-1,1]\) such that \(t(x) = 0\) for all but finitely many \(x\), and \(\sum_{x\in A}t(x) \in [-1, 1]\) for all \(A \subseteq X\).
		On morphisms, \(\weighted\) acts by summing up preimages, so for
\(g\colon X\rightarrow Y\), \(t\in \weighted X\) we have \(\weighted g(t)(y) = \sum_{x \in g^{-1}(y)} t(x)\).
\noindent The distance of two elements $s, t \in \weighted X$ is given by
			\(d(s, t) = \frac{1}{2}\bigvee_{f} \sum_{x\in X} f(x)s(x) - f(x)t(x)\)
			where the join ranges over all nonexpansive functions \(f \colon X \rightarrow [0,1]\).
			Then \((\weighted{-})^A\) coalgebras are \(([-1,1],+, 0)\)-weighted $A$-labelled transition systems. Behavioural equivalence instantiates to
                        weighted bisimilarity~\cite{DBLP:conf/birthday/Klin09}.

		\item\label{main:kant} We consider a variation of the Kantorovich functor $\kantorovich$~\cite{DBLP:journals/tcs/BreugelHMW07, DBLP:journals/jcss/AdamekMMU15}.
	We say that a probability measure $\mu$ on the Borel $\sigma$-algebra of $X$  is \emph{tight} if for every $\epsilon > 0$, there is a totally bounded subset $Y \subseteq X$ such that $\mu(X \setminus Y) < \epsilon$. The Kantorovich functor $\kantorovich$ maps a metric space  $(X,d)$ to the set of tight probability measures on $X$, equipped with the Kantorovich metric, defined as

	\(d_{\kantorovich X}(\mu, \nu) = \sup_{f} \left\{\int f \dif\mu - \int f \dif\nu \right\}\) for $\mu, \nu \in \kantorovich X$,
			where again \(f\) ranges over all nonexpansive maps \(X \rightarrow [0, 1]\).
			On morphisms, \(\kantorovich\) acts by measuring preimages, i.e. for $f \colon  X \rightarrow Y$ we have $\kantorovich f(\mu)(Y') = \mu(f^{-1}(Y'))$ for all $Y' \in \kantorovich(Y, d_Y)$.
			Continuous probabilistic transition systems are $\kantorovich(1 + {-})^A$-coalgebras for a finite set \(A\) of actions \cite{DBLP:journals/tcs/BreugelHMW07,DBLP:journals/tcs/BreugelW05} (so the term \emph{continuous} applies to the state space, not the system evolution), and behavioural equivalence instantiates to  probabilistic bisimilarity of continuous systems.

	\end{enumerate}
\end{example}

\begin{figure}[t]
\centering
\begin{tikzpicture}[node distance=2cm, auto]
				\node at (2.5, 1) (A)[circle, fill,inner sep=2pt] {};
				\node at (1.5, 0) (B)[circle, fill,inner sep=2pt] {};
				\node at (3.5, 0) (C)[circle, fill,inner sep=2pt] {};

				\node at (8.5, 1) (A1)[circle, fill,inner sep=2pt] {};
				\node at (7.5, 0) (B1)[circle, fill,inner sep=2pt] {};
				\node at (9.5, 0) (C1)[circle, fill,inner sep=2pt] {};

				\path[->]
					(A) edge node[left=2ex] {\(\frac{1}{2}\)} (B)
						edge node[right=2ex] {\(\frac{1}{2}\)} (C)
					(C) edge[loop, right, in=45,out=-45, min distance=4mm,looseness=10] node {1} (C)

					(A1) edge node[left=2ex] {\(\frac{1}{2} + \epsilon\)} (B1)
						 edge node[right=2ex] {\(\frac{1}{2} - \epsilon\)} (C1)

					(C1) edge[loop right, in=45,out=-45,min distance=4mm,looseness=10] node {1} (C1);
\end{tikzpicture}
 \caption{Probabilistic transition systems with behaviourally inequivalent root states.}
\label{fig:notBeq}
\end{figure}
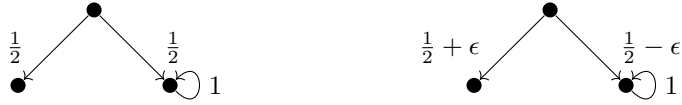

\noindent Consider the probabilistic transition systems depicted in Figure~\ref{fig:notBeq}.
If \(\epsilon > 0\), then the root states are not probabilistically bisimilar, as they have different probabilities of reaching a deadlock state. Still one would like to say that their difference in behaviour is small if~$\epsilon$ is small.  We will review formal definitions of such concepts  in Section~\ref{sec:v-ml}.

\section{Quantales and Quantale-Enriched Categories}

A central notion of our development is the notion of quantale, which is a joint generalization of
truth values and distances.
A \df{quantale}~\((\V,\bigvee,\otimes,k)\), more precisely a commutative and unital quantale, is a
complete lattice $\V$ that carries the structure of a commutative monoid
\((\V,\otimes,k)\), with~$\otimes$ called \df{tensor} and~$k$ called \df{unit}, such that for every $u \in \V$, the map $u \otimes - \colon
\V\to \V$ preserves suprema, which entails that

every $u \otimes - \colon \V \to \V$ has a right
adjoint $\hom(u,-) \colon \V \to \V$, characterized by the property
\(
  u \otimes v \le w \iff v \le \hom(u,w).
\)

We denote by $\top$ and $\bot$ the greatest and the least
element of a quantale respectively.
A quantale is \df{non-trivial} if $\bot\neq\top$, and \df{integral} if
\(\top = k\).

\begin{example}
\label{p:30}
  \begin{enumerate}
	\item\label{item:frame} Every frame (i.e.\ a complete lattice in which binary meets distribute over
	infinite joins) is a quantale with \(\otimes = \wedge\) and \(k = \top\).
	In particular, every finite distributive lattice is a quantale, prominently~\(\two\), the two-element lattice $\{\bot,\top\}$.
  \item Every left continuous \(t\)-norm \citep{AFS06} defines a quantale on the unit interval equipped with its natural order.
	\item \label{ex:metric-q} The previous clause (up to isomorphism) further specializes as follows:
	  \begin{enumerate}
	    \item The quantale \([0,\infty]_+ = ([0,\infty], \inf, +, 0)\) of non-negative real numbers with infinity, ordered by the greater or equal relation, and with tensor given by addition.

	    \item The quantale \([0,\infty]_{\max} = ([0,\infty], \inf, \max, 0)\) of non-negative real numbers with infinity, ordered by the greater or equal relation, and with tensor given by maximum.
	    \item \label{p:4} The quantale \([0,1]_\oplus = ([0,1], \inf, \oplus, 0)\) of the unit interval, ordered by the greater or equal order, and with tensor given by truncated addition.
	  \end{enumerate}
	  (Note that the quantalic order here is dual to the standard numeric order).

   \item Every commutative monoid \((M, \cdot, e)\) generates a quantale structure on \((\pow M, \bigcup)\), the free quantale on~\(M\).
    The tensor \(\otimes\) on \(\pow M\) is defined by
    \(
      A \otimes B = \{ a \cdot b \mid a \in A \text{ and } b \in B\},
    \)
    for all \(A,B \subseteq M\).
    The unit of this multiplication is the set \(\{e\}\).

   \item \label{p:29} For every quantale \(\V\) and every partially ordered set
  \(X\), the set of monotone maps \(\POST(X,\V)\) ordered pointwise becomes
  a quantale with tensor defined pointwise. For instance,

  \(\POST(2,[0,1]_\oplus)\) with discrete $2$ yields the quantale \([0,1]_\oplus^2\),
  and by replacing $2$ with the two-element chain \(0 \geq 1\) we obtain the quantale \(\calI([0,1]_\oplus)\) of
  non-empty closed subintervals of \([0,1]\) \citep{WS21}.
	\end{enumerate}
\end{example}

Category theory highlights preordered sets as $\two$-enriched categories. By replacing $\two$
with a quantale \(\V\), we enrich the relevant preorders with a quantitative extent:
A \df{\(\V\)-category} is pair \((X,a)\) consisting of a set $X$ and a map $a \colon X \times X \to \V$ that for all $x,y,z \in X$ satisfies the inequalities \(k \le a(x,x)\) and  \(a(x,y) \otimes a(y,z) \leq a(x,z)\).
A \(\V\)-category \((X,a)\) is \df{symmetric} if \(a(x,y) = a(y,x)\), for all \(x,y \in X\).

Every $\V$-category \((X,a)\) carries a \df{natural order} defined by
\(
  x \le y \text{ whenever } k \le a(x,y),
\)
which induces a faithful functor $\Cats{\V} \to \ORD$.
A \(\V\)-category is \df{separated} if its natural order is antisymmetric.

A \df{$\V$-functor} $f \colon (X,a) \to (Y,b)$ is a map
$f \colon X \to Y$ such that, for all $x,y \in X$,
\( a(x,y) \le b(f(x),f(y)).  \) \(\V\)-categories and \(\V\)-functors
form the category \(\Cats{\V}\); we denote by \(\Cats{\V}_\sym\)
and~\(\Cats{\V}_{\sym,\sep}\) the full subcategories of \(\Cats{\V}\)
determined by the symmetric and the symmetric separated
\(\V\)-categories, respectively.

\begin{example}

  \begin{enumerate}
  \item
      The category $\Cats{\two}$ is equivalent to the category $\ORD$ of \df{preordered sets} and monotone maps.
    \item  Metric, ultrametric and bounded metric spaces \`a la Lawvere \cite{Law73} can be seen as quantale-enriched categories:
      \begin{enumerate}
		  \item The category \(\Cats{[0,\infty]_+}_{\sym,\sep}\) is equivalent to the category \(\MET\) of generalized \df{metric spaces} and non-expansive maps.
		  \item The category \(\Cats{[0,\infty]_{\max}}_{\sym,\sep}\) is equivalent to the category $\UMET$ of generalized \df{ultrametric spaces} and non-expansive maps.
		  \item The category $\Cats{[0,1]_{\oplus}}_{\sym,\sep}$ is equivalent to the category $\BMET$ of \df{bounded-by-$1$ metric spaces} and non-expansive maps.
      \end{enumerate}

    \item Categories enriched in a free quantale \(\pow M\) on a monoid \(M\) (such as $M=\Sigma^*$ for some alphabet~$\Sigma$) can be interpreted as \df{labelled transition systems} with labels in~$M$: in a \(\pow M\)-category \((X,a)\), the objects represent the states of the system, and we can read \(m\in a(x,y)\) as an $m$-labelled transition from~\(x\) to~\(y\).

    \item Categories enriched in the quantale \(\calI([0,1]_\oplus)\) can be interpreted
    as spaces where to each pair of points a distance range is assigned.

  \end{enumerate}
\end{example}

\begin{table}[t]
\begin{center}
\begin{tabular}{l|>{\raggedright\arraybackslash}p{4cm}|>{\raggedright\arraybackslash}p{5cm}}
	General $\V$ & Qualitative ($\V=2$) & Quantitative ($\V=[0,1]_\oplus$) \\\hline
	$\V$-category & preorder & bounded-by-1 hemi-metric space\\
	symmetric $\V$-category & equivalence & bounded-by-1 pseudo-metric space\\
	$\V$-functor &  monotone map & non-expansive map\\
	initial $\V$-functor & order-reflecting\newline monotone  map & isometry\\
	L-dense $\V$-functor & monotone map that is surjective up to the induced equivalence & non-expansive map with dense \mbox{image}\\
	L-closure & closure under the induced equivalence & topological closure\\
\end{tabular}
\end{center}
\caption{$\V$-categorical notions in the qualitative and the
  quantitative setting. The prefix `pseudo' refers to absence of
  separatedness, and the prefix `hemi' additionally indicates absence
  of symmetry.}
\label{fig:sep}
\end{table}

The examples $\V=\two$ and $\V=[0,1]_\oplus$ are particularly
instructive, as they represent the most established qualitative and
quantitative aspects quantales aim to generalize. Table~\ref{fig:sep}
provides some instances of generic quantale-based concepts (either
introduced above or to be introduced presently) in these two cases,
for further reference.

There are several ways of constructing \(\V\)-categories from existing
ones.  Of special importance is the \(\V\)-category composed by a set
\(X\) equipped with the \df{initial structure}
\(a \colon X \times X \to \V\) w.r.t.\ a \df{structured cone}, i.e.\ a
family \((f_i \colon X \to \ftII{(X_i,a_i)})_{i\in I}\) of maps; that
is,

\(
  a(x,y) = \bigwedge_{i\in I}a_i(f_i(x),f_i(y)),
\)
for all $x,y\in X$.
A cone \(((X,a) \to (X_i,a_i))_{i \in I}\) is said to be \df{initial} (w.r.t.\  the forgetful functor \(\ftII{-} \colon \Cats{\V} \to \SET\)) if \(a\) is the initial structure w.r.t.\  the structured cone \((X \to \ftII{(X_i,a_i)})_{i \in I}\); a \(\V\)-functor is initial if it forms an initial cone.

  Every quantale $\V$ is a $\V$-category, which we also denote by \(\V\), with
  $\hom(x,y)$ being a hom-object over $x$ and $y$.
  For every set $S$, let the $S$-power~$\V^S$ of $\V$ be the set of all functions
  $h \colon S \to \V$ equipped with the $\V$-category structure $[-,-]$ given
  by a certain initial structure, concretely,
  \(
    [h,l] = \bigwedge_{s \in S} \hom(h(s),l(s)),
  \)
  for all $h,l \colon S \to \V$.

\begin{remark}\label{rem:sym-sep}
Let us briefly outline the connections between \(\Cats{\V}\), \(\Cats{\V}_\sym\)
and \(\Cats{\V}_{\sym,\sep}\), which for real-valued $\V$ correspond to
hemi-metric, pseudo-metric and metric spaces, respectively.

The category \(\Cats{\V}\) is topological over \(\SET\)~\cite{AHS90} and so
is \(\Cats{\V}_\sym\) (it inherits the initial cone construction); in particular,
both categories are complete and cocomplete. Moreover, the category
\(\Cats{\V}_\sym\) is a (reflective and) coreflective full subcategory of \(\Cats{\V}\).
The coreflector \((-)_s \colon \Cats{\V} \to \Cats{\V}_\sym\) is identity
on morphisms and sends every \((X,a)\) to its symmetrization,
the \(\V\)-category \((X,a_s)\) where \(a_s(x,y)= a(x,y) \wedge a(y,x)\) (keep in mind that in Examples~\ref{ex:metric-q} the order the dual of the numeric order).

Similarly, \(\Cats{\V}_{\sym,\sep}\) is a full reflective subcategory
of \(\Cats{\V}_\sym\) (but not topological). The relevant reflector
\((-)_q \colon \Cats{\V}_\sym \to \Cats{\V}_{\sym,\sep}\) quotients
every \(X\) by its natural preorder, which for symmetric~$X$ is an
equivalence.

\end{remark}

We will use \(\Cats{\V}_\sym\) as the main device for formalizing examples and
stating our main results, although most of these results can be meaningfully reinterpreted for \(\Cats{\V}\).

Quantale-enriched categories come equipped with a canonical closure operator.
A \(\V\)-functor \(m \colon M\to X\) is \df{L-dense} \citep{HT10} if for all \(\V\)-functors \(f,g \colon X\to \V\), \(f\cdot m=g\cdot m\)

implies \(f= g\).
We immediately notice that the composite of L-dense \(\V\)-functors is L-dense.

For a subset \(A\) of \(X\), the \df{L-closure} \(\overline{A}\) of
\(A\) in \((X,a)\) is the largest \(\V\)-subcategory of \((X,a)\) in which \(A\) is L-dense;
this can be explicitly computed as follows

\begin{displaymath}
 \overline{A} = \bigl\{x \in X \mid k \leq \textstyle\bigvee_{y \in A} a(x,y) \otimes a(y,x)\bigr\}.
\end{displaymath}

A subset $A \subseteq X$ of a $\V$-category $(X,a)$ is \df{closed} if $A = \overline{A}$, and \df{dense in $(X,a)$} if $\overline{A} = X$.
A function \(f \colon X \to Y\) between \(\V\)-categories $(X,a)$ and $(Y,b)$ is said to be \df{continuous} if \(f\overline{A} \subseteq \overline{fA}\), for every \(A \subseteq X\).

It is easy to see that the notion of L-closure is so designed that for metric-like
examples it coincides with the topological closure w.r.t.\ the corresponding
open-ball topology, and continuous functions are continuous in this topology.
For a preorder \((X,\leq)\), $A\subseteq X$ is closed
iff it is closed under the induced equivalence.
It is easy to check that every \(\V\)-functor is continuous.

\begin{proposition}
  \label{prop:Vs-closed}
  For every \(\V\)-category \(X\), \(\Cats{\V}(X,\V)\) is closed in \(\V^{\ftII{X}}\). For every symmetric \(\V\)-category \(X\), \(\Cats{\V}(X,\V_s)\) is closed in \(\V_s^\ftII{X}\).
\end{proposition}

Maps between \(\V\)-categories that preserve closure will be essential in Section~\ref{sec:expr}.

\section{Quantitative Coalgebraic Modal Logics}
\label{sec:v-ml}

We proceed to introduce a variant of (quantitative) coalgebraic
logic~\citep{Pat04,Sch08,CKP+11,KonigMikaMichalski18,WS20}, which in
particular follows the paradigm of interpreting modalities via
predicate liftings, in this case of $\V$-valued predicates.

Given a cardinal~\(\kappa\), a classical \(\kappa\)-ary predicate lifting for a
functor \(\ftF \colon \SET \to \SET\) is a natural transformation
\(
  \lambda \colon \SET(-,\two^\kappa) \to \SET(\ftF -,\two).
\)

For example, Kripke semantics of the modal logic \(K\) can be coached in terms
of the diamond modality \(\Diamond\), which with identify with the unary predicate
lifting $\Diamond_X(A) = \{B \subseteq X \mid A \cap B \neq \varnothing\}$ for
the (finite) powerset functor with (modulo the isomorphism $\mathcal{P} X\cong\SET(X,\two)$).

This notion of predicate lifting naturally extends to \(\Cats{\V}_\sym\)-functors.

\begin{defn}
	Given a cardinal \(\kappa\), a \(\kappa\)-ary \df{predicate lifting} for a functor \(\ftF \colon \Cats{\V}_\sym \to \Cats{\V}_\sym\) is a natural transformation of type
  \(
    \lambda\colon\Cats{\V}_\sym(-,\V_s^\kappa) \to \Cats{\V}_\sym(\ftF-,\V_s).
  \)

\end{defn}

\begin{remark}[Yoneda Lemma]
  \label{p:43}
  By the Yoneda lemma, every \(\kappa\)-ary predicate lifting for a functor \(\ftF \colon \Cats{\V}_\sym \to \Cats{\V}_\sym\)
  is completely determined by its action on the identity map on \(\V_s^\kappa\).
  By abuse of notation we will use \(\lambda\) to refer to the resulting morphism of type \(\ftF \V_s^\kappa \to \V_s\), as well as the corresponding predicate lifting if the type is clear from the context. The original predicate lifting is recovered by the assignment \(f \mapsto \lambda \cdot \ftF f\).
\end{remark}

For the sake of brevity, we restrict to unary predicate liftings ($\kappa=1$) henceforth.

The syntax of quantitative coalgebraic modal logic can now be defined by the grammar

\begin{flalign*}
  &&\phi\Coloneqq \top \mid \phi_1 \vee \phi_2 \mid \phi_1 \wedge \phi_2 \mid u \otimes \phi \mid \hom_s(u,\phi) \mid \lambda(\phi) && (u \in \V, \lambda \in \Lambda)
\end{flalign*}

where $\Lambda$ is a set of \emph{modalities}, which we
identify, by abuse of notation, with predicate liftings for a functor \(\ftF \colon \Cats{\V}_\sym \to \Cats{\V}_\sym\). We view all other connectives as propositional operators.
Let $\calL(\Lambda)$ be the set of modal formulas thus defined.

The semantics is given by assigning to each formula \(\phi \in \calL(\Lambda)\) and each coalgebra \(\alpha \colon X \to \ftF X\) the \emph{interpretation} of~$\phi$ over~$\alpha$,
the \(\V\)-functor \(\Sema{\phi} \colon X \to \V_s\) recursively defined as follows:

\begin{itemize}
	\item for \(\phi = \top\), we take \(\Sema{\top}\) to be the \(\V\)-functor given by the constant map into \(\top\);
	\item for an $n$-ary propositional operator~$p$, we put $\Sema{p(\phi_1, \ldots, \phi_n)} = p (\Sema{\phi_1}, \ldots, \Sema{\phi_n})$, with~$p$ interpreted using the lattice structure of~$\V$ and the \(\V\)-categorical structure \(\hom_s\) of \(\V_s\), respectively, on the right-hand side;
	\item for $\lambda \in \Lambda$, we put $\Sema{\lambda(\phi)} = \lambda (\Sema{\phi}) \cdot \alpha$.
\end{itemize}

Given a coalgebra \((X,\alpha)\), we denote the set of all maps of the form \(\Sema{\phi} \) by \(\Sema{\calL}\).

Interpreting~$\top$ as $\top$ and not as $k$ is essential for non-integral quantales,
for which the constant map with value~$k$ fails to be a~$\V$-functor.

\begin{proposition}
  \label{prop:form-sem-f}
  Let \(\Lambda\) be a set of predicate liftings for a functor \(\ftF \colon \Cats{\V}_\sym \to \Cats{\V}_\sym\), and
  let \(f \colon (X,\alpha) \to (Y, \beta)\) be a  homomorphism of \(\ftF\)-coalgebras.
  Then, for every formula \(\phi \in \calL(\Lambda)\), \(\Sema{\phi} = \Semb{\phi} \cdot f\).
\end{proposition}

The established approach to coalgebraic behavioural distances is to start with a set functor $\ftF\colon\SET\to\SET$ and obtain a $\Cats{\V}_\sym$-functor,as a lifting of~$\ftF$. In the quantalic setting, this approach may take the following shape.

Topological properties of $\Cats{\V}_\sym$ entail that every set\(\Lambda\)
of predicate liftings for a functor \(\ftF \colon \SET \to \SET\)
induces a functor \(\ftF^\Lambda \colon \Cats{\V}_\sym \to \Cats{\V}_\sym\),
known as the \df{Kantorovich lifting} of \(\ftF\) w.r.t.\  \(\Lambda\) \cite{BBKK18}.
More concretely, \(\ftF^\Lambda\)
sends a \(\V\)-category \((X,a)\) to the \(\V\)-category determined by the initial
structure on \(\ftF X\) w.r.t.\  the structured cone of all maps
\(\lambda(f) \colon \ftF X \to \ftII{\V_s}\) with \(\lambda \in \Lambda\) and \(f \colon (X,a) \to \V_s \in \Cats{\V}_\sym\).

As the name indicates,~$\ftF^\Lambda$ is indeed a \df{lifting} of \(\ftF\)
to \(\Cats{\V}_\sym\), that is, \(\ftII{-} \cdot \ftF^\Lambda = \ftF \cdot \ftII{-}\)
where \(\ftII{-} \colon \Cats{\V}_\sym \to \SET\) is the forgetful functor.

Every predicate lifting \(\lambda\in\Lambda\) for $\ftF$ becomes a  predicate lifting for the Kantorovich lifting~$\ftF^\Lambda$.

Kantorovich liftings are crucial prerequisites for existing
expressiveness results of quantitative coalgebraic logics for
\(\SET\)-functors (e.g.~\cite{WS20,WS21,KKK+21}).  However, it turns
out that the Kantorovich property can be usefully detached from the
notion of functor lifting:

\begin{defn}[Kantorovich Functor]
  Let \(\Lambda\) be a set of predicate liftings for a functor \(\ftF \colon \Cats{\V}_\sym \to \Cats{\V}_\sym\).
  The functor \(\ftF\) is \df{\(\Lambda\)-Kantorovich} if for every \(\V\)-category \(X\), the cone of all \(\V\)-functors
  \(
    \lambda(f) \colon \ftF X \to \V_s
  \)
  with \(\lambda \in \Lambda\) and \(f\in \Cats{\V}_\sym(X,\V_s)\) is initial.
\end{defn}
Clearly, every Kantorovich lifting \(F^\Lambda\) is \(\Lambda\)-Kantorovich.

\begin{example}\label{expl:expectation}

Recall the finite distribution functor $\mathcal{D} \colon \SET \to \SET$ from Example~\ref{exa:main}\eqref{main:prob}, and let~$\mathbb{E}$ be its \emph{expectation} predicate lifting:
given a map \(f\colon X \to [0,1]\), for every finite distribution \(\mu\) on \(X\), $\mathbb{E}_X(f)(\mu) = \sum_{x\in X} f(x) \mu(x)$.
The corresponding lifting~$\mathcal{D}^\mathbb{E}$ is the usual Kantorovich distance on distributions, restricted to finite distributions.
The closely related functor $\kantorovich$ from Example~\ref{exa:main}\eqref{main:kant} already lives in $\Cats{[0,1]_\oplus}_\sym$ and is not a lifting of any set functor;
however, $\kantorovich$ is \(\Lambda\)-Kantorovich for \(\Lambda = \{ \mathbb{E}\}\) where $\mathbb{E}$
is the (continuous) expectation predicate lifting:
given a non-expansive map \(f\colon X \to [0,1]\), for every tight probability distribution \(\mu\) on \(X\), \(\mathbb{E}_X(f)(\mu) = \int_X f(x)\;d\mu(x)\).
The functor~$\mathcal{D}^\mathbb{E}$ is a subfunctor of~$\kantorovich$, and the components of the associated inclusion natural transformation are initial.

The fact that \(\kantorovich\) is \(\{\mathbb{E}\}\)-Kantorovich makes it easy to verify that the functor \(\kantorovich(1 + {-})^A\) is \(\Lambda\)-Kantorovich for \(\Lambda = \{\mathbb{E}^{a,+1} \mid a \in A\}\), where given a non-expansive map \(f \colon X \to [0,1]\), for every \(A\)-indexed family \(l\) of tight probability measures on \(1+X\), \(\mathbb{E}^{a,+1}_X(f)(l)\) returns the expected value w.r.t\ \(l_a\) of the non-expansive map \(f^{+1}  \colon X+1 \to [0,1]\) that acts as \(f\) on the elements coming from \(X\) and sends the other element to \(1\).
Explicitly: \(\mathbb{E}^{a,+1}_X(f)(l) = \int_{1+X} f^{+1}(x)\;dl_a(x)\).
Consequently, the functor \(\mathcal{D}^\mathbb{E}(1+-)^A\) is also Kantorovich.
In fact, this functor is the Kantorovich lifting of the functor \(\mathcal{D}(1+-)^A \colon \SET \to \SET\) w.r.t\ \(\Lambda = \{\mathbb{E}^{a,+1} \mid a \in A\}\) where the predicate liftings \(\mathbb{E}^{a,+1}\) for \(\mathcal{D}(1+-)^A\) are constructed analogously to the predicate liftings \(\mathbb{E}^{a,+1}\) for \(\kantorovich(1+-)^A\).

\end{example}

\section{Behavioural and Logical Distances}
\label{sec:ld-and-bd}

Every functor \(\ftF \colon \Cats{\V}_\sym \to \Cats{\V}_\sym\) defines a notion of distance on an \(\ftF\)-coalgebra \((X,a,\alpha)\) that, similarly to behavioural equivalence, depends on all coalgebra homomorphisms, i.e.\ morphisms with domain \((X,a,\alpha)\) in the category $\Coalg{\ftF}$ of coalgebras for~$\ftF$.

\begin{defn}
  Let \(\ftF \colon \Cats{\V}_\sym \to \Cats{\V}_\sym\) be a $\V$-functor.
  The \df{behavioural distance} on an \(\ftF\)-coalgebra \((X,a,\alpha)\), denoted by \(bd_\alpha^\ftF\), is defined on all \(x,y \in X\) by

  \begin{align}\label{eq:bd-form}
    \bd_\alpha^\ftF(x,y) = \bigvee \{b(f(x),f(y)) \mid f \colon (X,a,\alpha) \to (Y,b,\beta) \in \Coalg{\ftF}\}.
  \end{align}

\end{defn}

If a \(\Cats{\V}_\sym\)-functor preserves initial morphisms, then the corresponding notion of behavioural distance is characterized by the coalgebra homomorphisms over identity maps.

In what follows, we denote by $1_X$ identity-on-objects $\V$-functors
\(1_X \colon (X,a) \to (X,b)\), assuming that $a\leq b$. These can extend to
coalgebra morphisms \(1_X \colon (X,a,\alpha) \to (X,b,\beta)\) under two
conditions, the first being commutativity of the diagram:

  \begin{displaymath}
\begin{tikzcd}[column sep=5em, row sep=4ex]
      (X,a) & \ftF (X,a) \\
      (X,b) & \ftF (X,b)
      \ar[from=1-1, to=1-2, "\alpha"]
      \ar[from=1-1, to=2-1, "1_X"']
      \ar[from=1-2, to=2-2, "\ftF 1_X"]
      \ar[from=2-1, to=2-2, "\ftF 1_X\cdot\alpha", dashed]
    \end{tikzcd}
  \end{displaymath}

which entails that $\beta$ must be \(\ftF 1_X \cdot \alpha\). The second condition
is that $\beta=\ftF 1_X \cdot \alpha$ must be a~$\V$-functor.

\begin{proposition}
  \label{prop:bd-formula}
  Let \((X,a,\alpha)\) be a coalgebra for a functor
  \(\ftF \colon \Cats{\V}_\sym \to \Cats{\V}_\sym\) that preserves
  initial morphisms. Then, we can equivalently
  restrict~\eqref{eq:bd-form} to morphisms of the form $f=1_X$:

  \begin{displaymath}
    \bd_\alpha^\ftF (x,y) = \bigvee \{b(x,y) \mid (X,b)\in\Cats{\V}_\sym, \ftF 1_X \cdot \alpha\in\Cats{\V}_\sym((X,b),\ftF(X,b)),a\leq b\}.
  \end{displaymath}

\end{proposition}

\begin{remark}
  \label{p:23}

  Since the forgetful functor \(\ftII{-} \colon \Cats{\V}_\sym \to \SET\) is topological, the elements of its fiber over a set \(X\) that are greater or equal than an element \((X,a)\) form a complete lattice \(\{(X,b) \in \Cats{\V}_\sym \mid a \leq b\}\).
  Moreover, for every \(\ftF\)-coalgebra \((X,a,\alpha)\), the endofunction on this complete lattice that sends a \(\V\)-category \((X,b)\) to the \(\V\)-category given by the initial structure on \(X\) w.r.t.\  the structured map \(\ftII{\ftF 1_X \cdot \alpha} \colon X \to \ftII{\ftF(X,b)}\) is monotone.
  Therefore, by Tarski's fixed point theorem this map has a greatest fixed point.
  By Proposition~\ref{prop:bd-formula}, if \(\ftF\) preserves initial morphisms, then this greatest fixed point is precisely the behavioural distance on \((X,a,\alpha)\).
  In particular, it follows that \(\beta \colon (X,bd_\alpha^\ftF) \to \ftF(X,bd_\alpha^\ftF)\) is a \(\V\)-functor.
  Furthermore, if \(\ftF\) is a lifting of a functor \(\ftG \colon \SET \to \SET\), then the behavioural distance on an \(\ftF\)-coalgebra \((X,a,\alpha)\) is given by the greatest \(\V\)-categorical structure on \(X\) that makes the \(\ftG\)-coalgebra \(\ftII{\alpha} \colon X \to \ftG X\) an \(\ftF\)-coalgebra.
  This is in line with the notion of behavioural distance based on liftings of \(\SET\)-functors (e.g.~\cite{BBKK18,KKK+21}).
\end{remark}

\noindent Behavioural distance is invariant under coalgebra homomorphisms:

\begin{proposition}
  \label{prop:bd-stable}
  Let \(\ftF \colon \Cats{\V}_\sym \to \Cats{\V}_\sym\) be a functor and \(f \colon (X,a,\alpha) \to (Y,b,\beta)\) a coalgebra morphism of \(\ftF\)-coalgebras.
  Then, for all \(x,y \in X\), \(bd_\alpha^\ftF (x,y) = bd_\beta^\ftF (f(x),f(y))\).
\end{proposition}

Coalgebraic modal logic complements behavioural distance with logical distance.

\begin{defn}
  Let \(\Lambda\) be a set of predicate liftings for a functor \(\ftF \colon \Cats{\V}_\sym \to \Cats{\V}_\sym\).
  The \df{logical distance} on an \(\ftF\)-coalgebra \((X,a,\alpha)\), denoted by \(ld_\alpha^\Lambda\), is the initial structure on \(X\) w.r.t.\  the structured cone of all maps
  \(
    \Sema{\phi} \colon X \to \ftII{(\V,\hom_s)}
  \)
  with \(\phi \in \calL(\Lambda)\).
  More explicitly, for all \(x,y \in X\),
  \begin{displaymath}
    ld_\alpha^\Lambda(x,y) = \bigwedge \{\hom_s(\Sema{\phi}(x),\Sema{\phi}(y)) \mid \phi \in \calL(\Lambda)\}.
  \end{displaymath}
\end{defn}

It follows immediately from Proposition~\ref{prop:form-sem-f} that logical distance is also invariant under coalgebra homomorphisms.

The remainder of the paper is devoted to establishing criteria under which
logical distance and behavioural distance coincide.

  Recall that a (quantitative) coalgebraic logic is \df{adequate} if for every
  \(\ftF\)-coalgebra \((X,\alpha)\), \(bd_\alpha^\ftF \leq ld_\alpha^\Lambda\),
  and \df{expressive} if \(ld_\alpha^\Lambda \leq bd_\alpha^\ftF\), for every \(\ftF\)-coalgebra~\((X,\alpha)\).

The former property is relatively easy.
\begin{theorem}
  \label{p:25}
  Let \(\Lambda\) be a set of predicate liftings for a functor \(\ftF \colon \Cats{\V}_\sym \to \Cats{\V}_\sym\).
  Then, the coalgebraic logic \(\calL(\Lambda)\) is adequate.
\end{theorem}

\section{Expressivity of Quantitative Coalgebraic Modal Logic}
\label{sec:expr}

We fix a functor \(\ftF \colon \Cats{\V}_\sym \to \Cats{\V}_\sym\) and
a set~\(\Lambda\) of predicate liftings for~$\ftF$ throughout this
section.

The following lemma is related to the Knaster-Tarski proof principle
for expressiveness identified in work on codensity liftings
\cite[Theorem~IV.5]{KKK+21}:

\begin{lemma}
  \label{lem:ld-le-bd}
  Let \((X,a,\alpha)\) be an \(\ftF\)-coalgebra.
  If the cone of all \(\V\)-functors \({\lambda(f) \colon \ftF (X,ld_\alpha^\Lambda) \to \V_s}\)
  with \(\lambda \in \Lambda\) and \(f \in \Sema{\calL(\Lambda)} \) is initial, then \(\calL(\Lambda)\) is expressive.
\end{lemma}

Kantorovich functors come with a natural proof method for Lemma~\ref{lem:ld-le-bd}.
By definition, the cone of all \(\V\)-functors from \((X,ld_\alpha^\Lambda)\) to \(\V_s\) that are interpretations of formulas of \(\calL(\Lambda)\) is initial. Roughly speaking, one wishes to conclude from this fact that one can \emph{approximate} every \(\V\)-functor  \((X,ld_\alpha^\Lambda)\to\V_s\) by interpretations of formulas; then, to show Lemma~\ref{lem:ld-le-bd} we just need to guarantee that predicate liftings preserve \emph{approximations}.  We formalize this approach using closure operators.
We begin by introducing some notation.

Given a \(\V\)-functor \(i \colon Y \to X\) and a set \(A \subseteq \Cats{\V}_\sym(X,\V_s)\),
we denote by \(A \cdot i\) the set \(\{ f \cdot i \mid f \in A\}\), by \(\Lambda(A)\)
the set \(\{\lambda(f) \mid f \in A, \lambda \in \Lambda\}\) and by \(\ftII{A}\) the set \(\{ \ftII{f} \mid f \in A\}\) (of maps).

It will be convenient to encapsulate the propositional part of the logic algebraically:

\begin{defn}\label{def:pp-algebra}
  Let \(X\) be a \(\V\)-category.
  A subset \(A\) of \(\V_s^X\) is a \df{propositional algebra} if it contains the  \(\V\)-functor that is constantly \(\top\), and is closed under the operations \(\wedge\), \(\vee\), \(\hom_s(u,-)\), and \(u \otimes -\), for every \(u \in \V\).
\end{defn}

In particular, given an $\ftF$-coalgebra \((X,a,\alpha)\),
\(\Sema{\calL(\Lambda)} \subseteq \V_s^{(X,\,ld_\alpha^\Lambda)}\)
is a propositional algebra.

\begin{defn}[\(\V_s\)-closure]
	Given a set \(X\), a \df{\(\V_s\)-closure operator} is a family \(\cC = (\cC_X)_{X \in \SET}\) of
	closure operators on \(\SET(X,\V)\)  (i.e.\ operators \(\cC_X\colon\pow(\SET(X,\V))\to\pow(\SET(X,\V))\)  satisfying the standard \emph{extensiveness},
	\emph{monotonicity} and \emph{idempotence} laws) such that
	for every symmetric \(\V\)-category \((X,a)\), \(\Cats{\V}_\sym((X,a),\V_s)\subseteq\SET(X,\V)\) is closed w.r.t.\  \(\cC_X\).
	When no ambiguities arise, we write~\(\cC(A)\) instead of \(\cC_X(A)\). As usual, a subset \(A\subseteq\SET(X,\V)\) is
	\df{\(\cC\)-dense} if \(\cC(A)=\SET(X,\V)\).
\end{defn}

A \(\V_s\)-closure operator \(\cC\) is equivalently given by a family
\(\overline{\cC} = (\overline{\cC}_X)_{X \in \Cats{\V}_\sym}\) of closure
operators on \(\Cats{\V}_\sym(X,\V_s)\) such that for all \(A \subseteq \Cats{\V}_\sym(X,\V_s)\) and \(B \subseteq \Cats{\V}(Y,\V_s)\), where~$|Y|=|X|$, if \(\ftII{A} = \ftII{B}\) then \(\ftII{\overline{\cC}_X(A)} = \ftII{\overline{\cC}_Y(B)}\).
We make no distinction between \(\cC\) and \(\overline{\cC}\),
e.g.\ we apply $\cC$ also to $A\subseteq\V_s^X=\Cats{\V}_\sym(X,\V_s)$.

\begin{example}
The following closure operators on \(\SET(X,\V)\) are \(\V_s\)-closure operators:

  \begin{enumerate}
    \label{p:27}
    \item the identity operator \(\cId_X\);
    \item the operator \(\cL^{\V_s}_X\) that sends every set to its L-closure in the \(\V\)-category \(\V_s^X\);

    \item the closure operator \(\ccInf_X\) that sends every set \(A\) to its closure under codirected infima and finite suprema;

    \item the closure operator \(\cInf_X\) that sends every set \(A\) to its closure under infima;

    \item \label{p:28} The closure operator \(\cFun_X\) that sends every set \(A\) to \(\ftII{\Cats{\V}_\sym(X_A,\V_s)}\), where \(X_A\) denotes the \(\V\)-category determined by the initial structure with respected to the structured cone of all maps \(f \colon X \to \ftII{\V_s}\) with \(f \in A\) (\(\cFun\) is in fact the closure operator of a Galois connection, relating to recent work by Beohar et al.~\cite{BeoharEA22}).

  \end{enumerate}
\end{example}

While \(\cId\) is the \emph{least} \(\V_s\)-closure operator, somewhat
less trivially \(\cFun\) is the \emph{greatest} one (and hence induces
the \emph{weakest} notion of density).

The next result connects initiality, closure and density.

\begin{proposition}
  \label{lem:closure-dense}
Let \(\cC\) be a \(\V_s\)-closure operator.
  For every \(\V\)-category \(X\) and \(A\subseteq\V_s^X\),
  \begin{enumerate}
    \item \label{p:44} if \(A\) is \(\cC\)-dense then \(A\) is initial; for $\cC=\cFun$,
    the converse holds as well;
    \item \label{p:42} if \(\cC(A)\) is initial then \(A\) is initial.
  \end{enumerate}
\end{proposition}

\begin{defn}
  Let \(\cC\) be a \(\V_s\)-closure operator.
  A predicate lifting \(\lambda\in\Lambda\) is \df{\(\cC\)-continuous} if every  \(X\)-component \(\lambda_X \colon \V_s^X \to \V_s^{\ftF X}\) is continuous w.r.t.\  \(\cC_X\) and \(\cC_{\ftF X}\).
\end{defn}

\begin{example}
  \label{p:35} A predicate lifting~$\lambda$ is\vspace{-0.5em}
  \begin{enumerate}
    \item always \(\cId\)-continuous;
    \item  \(\cL\)-continuous iff its components are L-continuous;
    \item \(\ccInf\)-continuous iff its components preserve codirected infima and finite suprema;
    \item  \(\cInf\)-continuous iff its components preserve all infima.

  \end{enumerate}
\end{example}

It is easily verified that if a predicate lifting \(\lambda\) for a \(\{\lambda\}\)-Kantorovich functor is \(\cFun\)-continuous, then it preserves initial cones.
Thus, \(\cFun\)-continuity is a very strong assumption, which by Lemma~\ref{lem:ld-le-bd} entails that \(\calL(\{\lambda\})\) is expressive.
In order to obtain expressivity results for coalgebraic logics under weaker
assumptions, we will consider situations where
\(\cFun\)-density can be equivalently described as \(\cC\)-density for more
suitable \(\V_s\)-closure operators~\(\cC\).

\begin{defn}
  Let \(\calI\) be a class of symmetric
  \(\V\)-categories.  A \(\V_s\)-closure operator \(\cC\)
  \df{characterizes initiality} on \(\calI\) if for every
  \(\V\)-category \(X \in \calI\) and every propositional algebra
  \(A \subseteq \V_s^X\), \(A\) is \(\cC\)-dense iff \(A\) is \(\cFun\)-dense.
\end{defn}
Characterization of initiality for a given class~$\calI$ depends only
on the quantale and the closure operator, and may be seen as a form of
Stone-Weierstraß property; we will give general Stone-Weierstraß
theorems for some classes of quantales in Section~\ref{sec:SW-th}. In
most of these, $\calI$ will be the class of finite symmetric
$\V$-categories (and in one case, the class of totally bounded
pseudometric spaces).

We introduce next a key technical definition.

\begin{defn}
  Let \(\cC\) be a \(\V_s\)-closure operator.
  A cocone \((i \colon X_i \to X)_{i \in \calI}\) of initial $\V$-functors in \(\Cats{\V}_\sym\) \df{coreflects \(\cC\)-density} if the cone \((- \cdot i \colon \V_s^X \to \V_s^{X_i})_{i \in \calI}\) \df{reflects \(\cC\)-density}; that is, for every \(A \subseteq \V_s^X\), if \(A \cdot i\) is \(\cC\)-dense for every \(i \in \calI\), then \(A\) is \(\cC\)-dense.
\end{defn}

Since by Lemma~\ref{lem:closure-dense}, \(\cFun\)-density is equivalent to
initiality, we refer to coreflection of \(\cFun\)-density also as \df{coreflection of initiality}.

\begin{example}\label{expl:corefl}
  \begin{enumerate}
    \item An initial \(\V\)-functor coreflects \(\cId\)-density iff it is L-dense.
    \item\label{item:tb} A classical $1$-bounded metric space \((X,d)\) is totally bounded if for every \(u > 0\) there is a finite set \(X_u\) such that for every \(x \in X\) there is \(y \in X_u\) so that \(d(x,y) < u\).
      It can be shown that for every totally bounded metric space \((X,d)\) the cocone of embeddings of finite subspaces coreflects \(\cL\)-density.

    \item It follows from \cite[Lemma~1.10(4)]{HT10} that every jointly L-dense directed cocone coreflects initiality.
          In particular, every directed colimit coreflects initiality.
  \end{enumerate}
\end{example}

Since in general we can only replace \(\cFun\)-density with \(\cC\)-density in a restricted class of \(\V\)-categories, our results will depend on functors that are compatible with such a class.

\begin{defn}
  A class~\(\calI\)  of symmetric \(\V\)-categories \df{coreflects initiality under} a functor~\(\ftF\) if for every \(\V\)-category \(X\) there is a cocone \((i \colon X_i \to X)_{i \in \calI_X}\) of initial \(\V\)-functors such that \(X_i \in \calI\) and the cocone \((\ftF i \colon \ftF X_i \to \ftF X)_{i \in \calI_X}\) coreflects initiality.

\end{defn}

\begin{example}
  \begin{enumerate}
  \item The class
    \(\Cats{\V}_\sym\) coreflects initiality under every \(\Cats{\V}_\sym\)-functor.
  \item The class of all finite symmetric \(\V\)-categories coreflects
    initiality under Kantorovich liftings of finitary \(\SET\)-functors
    to \(\Cats{\V}_\sym\).
  \end{enumerate}
\end{example}

\begin{proposition}\label{prop:coref-ini-preserve}
  Let \(j \colon \ftG\to \ftF\) be a natural transformation between
  \(\Cats{\V}_\sym\)-functors such that each component of~\(j\) is
  initial and L-dense.  If a class~\(\calI\) of symmetric
  $\V$-categories coreflects initiality under~\(\ftG\), then~\(\calI\)
  coreflects initiality under~\(\ftF\).
\end{proposition}
\begin{corollary}
  If $\ftG$ is a Kantorovich lifting of a finitary set functor and
  \(j \colon \ftG\to \ftF\) is a natural transformation between
  \(\Cats{\V}_\sym\)-functors such that each component of~\(j\) is
  initial and L-dense, then the class of all finite symmetric
  \(\V\)-categories coreflects initiality under~$\ftF$.
\end{corollary}
\noindent (We note that this observation generalizes the use of
finitarily separable lax extensions~\cite{WS20}.)

\noindent We are now ready to present our main result:

\begin{theorem}[Quantitative Coalgebraic Hennessy-Milner theorem]
  \label{thm:main}
  Let \(\ftF \colon \Cats{\V}_\sym \to \Cats{\V}_\sym\) be
  \(\Lambda\)-Kantorovich, let \(\cC\) be a \(\V_s\)-closure operator,
  and let \(\calI\) be a class of symmetric \(\V\)-categories such
  that~\(\calI\) coreflects initiality under~\(\ftF\) and \(\cC\)
  characterizes initiality on \(\calI\). If every predicate lifting in
  \(\Lambda\) is \(\cC\)-continuous, then \(\calL(\Lambda)\) is
  expressive.
\end{theorem}
\noindent We note that there is a balance to be struck in the choice
of~$\calI$: The larger~$\calI$ is, the more functors one finds under
which~$\calI$ coreflects initiality, but the harder it is to establish
characterization of initiality in~$\calI$. We tackle the latter issue
next.

\section{Stone-Weierstra\ss{}-Type Theorems}
\label{sec:SW-th}
We now develop some usage scenarios for Theorem~\ref{thm:main}. As a warm up,
we have

\begin{theorem}\label{thm:sw-finite}
	Let \(\V\) be a finite quantale. Then
	the \(\V_s\)-closure operator \(\cId\) characterizes initiality on finite symmetric   \(\V\)-categories.
\end{theorem}

(Note that for $\V=\two$, this is essentially  the well-known functional completeness of Boolean logic.)
Hence, by instantiating Theorem~\ref{thm:main}, we obtain
\begin{corollary}
  \label{cor:hm-finite}
  Let \(\V\) be a finite quantale,  and let \(\ftF \colon \Cats{\V}_\sym \to \Cats{\V}_\sym\) be a \(\Lambda\)-Kantorovich functor that admits as an L-dense subfunctor a lifting of a finitary \(\SET\)-functor.
  Then the coalgebraic logic \(\calL(\Lambda)\) is expressive.
\end{corollary}

Most remarkably, Corollary~\ref{cor:hm-finite} allows us to derive expressivity for
many-valued logics beyond~$\two$, but also it allows us to relate our present formulation of expressivity
to the one standardly adopted in the properly qualitative case (i.e.\ with $\V=\two$)~\cite{Pat04,Sch08}.
To that end, recall that a set of predicate liftings \(\Lambda\) for a functor \(\ftF \colon \SET \to \SET\) is
\df{separating} if for every set \(X\) the cone of all maps \(\lambda(f) \colon \ftF X \to \two\)
with \(\lambda \in \Lambda\) and \(f \colon X \to \two \) is mono. For a
$\SET$-coalgebra~\((X,\alpha)\), let us denote by \(beq_\alpha\) the standard notion of behavioural
equivalence, as explained in preliminaries. Let $\EQ=\Cats{\two}_\sym$ (the category of equivalence relations).

The following result generalizes \cite[Theorem~11]{MV15} (which applies only to functor liftings that arise from lax extensions, in particular requires modalities to be monotone):

\begin{theorem}
  \label{thm:kl-discrete}
  Let \(\ftF^\Lambda \colon \EQ \to \EQ\) be a Kantorovich lifting that preserves discrete equivalence relations.
  Then, for every \(\ftF\)-coalgebra \((X,\alpha)\), \(bd_\alpha^{\ftF^\Lambda} = beq_\alpha\).
\end{theorem}

We now can recover the general expressivity result~\cite{Pat04,Sch08} for $\V=\two$ as
a direct consequence of Corollary~\ref{cor:hm-finite} and Theorem~\ref{thm:kl-discrete}.

\begin{theorem}\label{thm:lutz-expr}
Let \(\Lambda\) be a separating set of predicate liftings for a finitary functor \(\ftF \colon {\SET \to \SET}\).
Then, the coalgebraic logic \(\calL(\Lambda)\) is expressive; that is, if two states in an \(\ftF\)-coalgebra are logically indistinguishable, then they are behaviourally equivalent.
\end{theorem}

Next, we obtain a characterization of $\cL$-density. Recall that an
element $x$ of an ordered set is \emph{way above} an element~$y$ if
whenever $y\ge\bigwedge A$ for a codirected set~$A$, then $x\ge a$ for some $a\in A$.

\begin{theorem}\label{thm:sw-L}
Suppose that $\V$ satisfies the condition
\begin{align}\label{eq:k-decomp}
k=\bigvee\{u \otimes u \mid u \in \V \text{ and for all \(v\in\V\), }\hom(u,v)\text{ is way above }v\}.
\end{align}
Then \(\cL\) characterizes initiality on finite symmetric \(\V\)-categories.
\end{theorem}

      We thus obtain the following generalization at quantalic level of the previous
      coalgebraic Hennessy-Milner theorem for finitarily separable \([0,1]_\oplus\)-lax
      extensions~\cite[Corollary 8.6]{WS20}:

\begin{corollary}\label{cor:hm-L}
Let \(\V\) be a quantale satisfying~\eqref{eq:k-decomp}, and let \(\ftF \colon \Cats{\V}_\sym \to \Cats{\V}_\sym\) be a \(\Lambda\)-Kantorovich functor that admits as an L-dense subfunctor a lifting of a finitary \(\SET\)-functor.
If every predicate lifting in~\(\Lambda\) is \(\cL\)-continuous, then \(\calL(\Lambda)\) is expressive.
\end{corollary}
Specifically, besides allowing more general quantales, we drop the assumptions that the modalities in~$\Lambda$ are monotone
and that~$\ftF$ is a lifting of a \(\SET\)-functor, and we weaken the assumption of non-expansiveness of predicate liftings to L-continuity.

The following fact sometimes allows us to enlarge the
      class on which initiality is characterized:

\begin{proposition}\label{prop:charc-ini}
  Let \(\cC\) be a \(\V_s\)-closure operator that characterizes initiality on \(\calI\), and let \(\calJ\) be a class of symmetric \(\V\)-categories.
  If for every \(X \in \calJ\) there is a cocone \((i \colon X_i \to X)_{i \in \calI_X}\) of initial \(\V\)-functors that coreflects \(\cC\)-density and \(X_i \in \calI\) for each \(i \in \calI_X\), then \(\cC\) characterizes initiality on \(\calJ\).
\end{proposition}
In particular, we obtain by Example~\ref{expl:corefl}.\ref{item:tb}
that for $\V=[0,1]_{\oplus}$,~$\cL$ characterizes initiality on
totally bounded pseudometric spaces, so we can, in this case, further
relax the assumptions of Corollary~\ref{cor:hm-L} as follows.
\begin{defn}
  A functor~$\ftF\colon\Cats{\V}_\sym\to\Cats{\V}_\sym$, for
  $\V=[0,1]_{\oplus}$, is \emph{totally bounded} if for every
  symmetric $\V$-category~$X$ and every $t\in\ftF X$, there exists a
  totally bounded $X_0\subseteq X$ and $t'\in\ftF X_0$ such that
  $t=\ftF i(t')$ where~$i$ is the inclusion $X_0\to X$.
\end{defn}
\begin{corollary}\label{cor:hm-L-tb}
Let \(\ftF \colon \Cats{[0,1]_\oplus}_\sym \to \Cats{[0,1]_\oplus}_\sym\) be a \(\Lambda\)-Kantorovich functor that admits an L-dense totally bounded subfunctor.
If every predicate lifting in \(\Lambda\) is \(\cL\)-continuous, then~\(\calL(\Lambda)\) is expressive.
\end{corollary}
\noindent We conclude with some variants employing order-theoretic
closure operators:

\begin{theorem}\label{thm:sw-cInfsup}
  Let \(\V\) be a quantale such that for every \(u \in \V\) the map \(u \otimes -\) preserves codirected infima.
  The closure operator \(\ccInf\) characterizes initiality on  finite symmetric \(\V\)-categories.
\end{theorem}
Notice that the assumption on~$\V$ in the above theorem is satisfied
in particular when~$\V$ is a frame (Example~\ref{p:30}.\ref{item:frame}).

\begin{corollary}\label{cor:hm-frame}
  Let \(\V\) be a quantale such that for every \(u \in \V\) the map \(u \otimes -\) preserves codirected infima, and let \(\ftF \colon \Cats{\V}_\sym \to \Cats{\V}_\sym\) be a \(\Lambda\)-Kantorovich functor that admits as an L-dense subfunctor a lifting of a finitary \(\SET\)-functor.
  If every predicate lifting in \(\Lambda\) preserves codirected infima and finite suprema, then the coalgebraic logic \(\calL(\Lambda)\) is expressive.
\end{corollary}

\begin{theorem}\label{thm:sw-cInf}
  Let \(\V\) be a completely distributive quantale such that for every \(u \in \V\) the map \(u \otimes -\) preserves codirected infima.
  The closure operator \(\cInf\) characterizes initiality on  finite symmetric \(\V\)-categories.
\end{theorem}
\begin{corollary}
  Let \(\V\) be a   completely distributive quantale such that for every \(u \in \V\) the map \(u \otimes -\) preserves codirected infima, and let \(\ftF \colon \Cats{\V}_\sym \to \Cats{\V}_\sym\) be a \(\Lambda\)-Kantorovich functor that admits as an L-dense subfunctor a lifting of a finitary \(\SET\)-functor.
  If every predicate lifting in \(\Lambda\) preserves all infima, then the coalgebraic logic \(\calL(\Lambda)\) is expressive.
\end{corollary}

\section{Examples}

Let us revisit the clauses mentioned in Example~\ref{exa:main} and their variants.

\ref{main:lts}. (\emph{Labelled/metric transition systems}): As such,
LTS are a purely qualitative system type, so we vary the setup and
instead consider a form of \emph{metric transition
  systems}. Specifically, we assume that every state is labelled with
a number in the unit interval, and for simplicity, we assume that
there is only one transition label. Moreover, we restrict to finite
branching. We thus consider coalgebras for the functor
$\ftF = [0,1]\times\pow_\omega$, where $\pow_\omega$ denotes the
finite powerset functor. Also, we aim for an example based on
\emph{ultrametrics}, so we make $[0,1]$ into a quantale
$\V=[0,1]_{\max}$. We define a nullary predicate lifting~$o$ (formally accommodated
as a unary predicate lifting that ignores its argument) by
$o_X(r,S)=r$, and a unary predicate lifting~$\Diamond$ by
\begin{equation*}
  \Diamond_X(f)(r,S)=\textstyle\bigvee_{x\in S} f(x).
\end{equation*}
We write $\bar\ftF$ for the Kantorovich lifting of $\ftF$ under the
set~$\Lambda$ of these modalities. We then obtain that the coalgebraic
logic $\calL(\Lambda)$ is expressive, via
Corollary~\ref{cor:hm-frame}: $\bar\ftF$ is itself a lifting of a
finitary set functor, $o$ trivially preserves all infima and suprema,
and one checks easily that~$\Diamond$ preserves codirected infima and
finite suprema.

\ref{main:para}. (\emph{$\rotatebox{45}{\scalebox{.75}{$\Box$}}$-valued powerset}): The set $\V=\{\bot,\mathsf{N},\mathsf{B},\top\}$ is an example
of a frame quantale (Example~\ref{p:30}). It is finite, and hence our general
Corollary~\ref{cor:hm-finite} applies to it, but not the previously existing
expressivity theorem for quantale-valued distances~\cite{WS21}, for this quantale is not a value-quantale. The
induced logic is a paraconsistent four-valued logic with \(\calL(\Lambda)\) instantiated
as follows:

\begin{flalign*}
  &&\phi\Coloneqq \top \mid \phi_1 \vee \phi_2 \mid \phi_1 \wedge \phi_2 \mid u  \mid \phi\equiv u \mid \lambda(\phi) && (u \in \V, \lambda \in \Lambda)
\end{flalign*}

where $\Lambda$ can be defined in different ways, and the expressivity result
remains true for any such choice. In practice, one usually considers further logical
connectives, which, of course, leaves expressivity intact. A natural extension of paraconsistent
logics with a box modality was previously considered by Rivieccio et al~\cite{RivieccioEA17},
who interpreted it over generalized Kripke frames, which are precisely $\paraP$-coalgebras.
They considered two ways of interpreting the box modality, which in our case
correspond to two predicate liftings: ${\Box_{\supset}}_X(f)(g) = \bigwedge_{x\in X}\hom(g(x),f(x))$
and ${\Box_{\to}}_X(f)(g) = \bigwedge_{x\in X}\hom(g(x),f(x))\land\hom(\neg f(x),\neg g(x))$
where $\neg$ is the paraconsistent negation: $\neg\top = \bot$, $\neg\bot = \top$,
$\neg\mathsf{B}=\mathsf{B}$ and $\neg\mathsf{N}=\mathsf{N}$.

\ref{main:prob} (\emph{Discrete probabilistic transition systems}):
As mentioned  already in Example~\ref{expl:expectation}, the functor \(\mathcal{D}^\mathbb{E}(1+-)^A\) is the Kantorovich lifting of the finitary functor \(\mathcal{D}(1+-)^A\) w.r.t\ \(\Lambda = \{\mathbb{E}^{a,+1} \mid a \in A\}\).
Moreover, it is easy to see that every predicate lifting in \(\Lambda\) gives rise to an \(\cL\)-continuous predicate lifting for \(\mathcal{D}^\mathbb{E}(1+-)^A\).
Hence, we obtain by Corollary~\ref{cor:hm-L} that quantitative probabilistic modal logic -- the coalgebraic modal logic generated by the expectation modality -- is expressive~\cite{DBLP:journals/tcs/BreugelHMW07, DBLP:journals/jcss/AdamekMMU15}.

\ref{main:weight}. (\emph{Weighted transition systems with negative weights}):

The functor \(\weighted\)

defining our variant of weighted transition systems is Kantorovich for
the set~$\Lambda$ of (\emph{non-monotone}) predicate liftings \(\langle a \rangle^{+r}\), for  \(a\in A\) and \(r \in \mathbb{R}\), defined by
\[
	\sem{\langle r\rangle}(f)(t) =\min\Bigl\{1, \max\Bigr\{0, r + \frac{1}{2}\sum_{x\in X} f(x)t(a)(x)\Bigr\}\Bigr\}
      \]
      Since~$\weighted$ is, moreover, a lifting of a finitary set
      functor, we obtain by Corollary~\ref{cor:hm-L} that the coalgebraic modal logic $\calL(\Lambda)$ is expressive.

\ref{main:kant}. (\emph{Continuous probabilistic transition systems}):
Our variant~$\kantorovich$ of the Kantorovich functor admits, by definition, a totally bounded L-dense subfunctor that assigns to a space~$X$ the set of all Borel distributions on~$X$ with totally bounded support.
Hence, as \(A\) is finite, it follows that the functor \(\kantorovich(1 + {-})^A\) admits a totally bounded L-dense subfunctor.
Furthermore, as noted already in Example~\ref{expl:expectation}, the functor \(\kantorovich(1+{-})^A\) is \(\Lambda\)-Kantorovich, where \(\Lambda = \{\mathbb{E}^{a,+1} \mid a \in A\}\), and it is easy to see that every predicate lifting in \(\Lambda\) is \(\cL\)-continuous.
Therefore, we obtain by Corollary~\ref{cor:hm-L-tb} that the coalgebraic modal logic \(\calL(\Lambda)\) is expressive, thus essentially recovering expressiveness of quantitative probabilistic modal logic on continuous probabilistic transition systems~\cite{DBLP:journals/tcs/BreugelW05,DBLP:journals/tcs/BreugelHMW07}.

\section{Conclusions and Further Work}

We have presented a quantitative Hennessy-Milner theorem in
coalgebraic and quantalic generality, covering behavioural distances
on a wide range of system types. Notably, our results apply to
functors on metric spaces that fail to be liftings of any set functor,
such as the (tight) Borel distribution functor. A key factor in the
technical development was the interplay between notions of density on
the one hand, and initiality of cones in the topological category of
generalized metric spaces taking values in a quantale~$\V$
(\emph{$\V$-categories}) on the other hand. We have illustrated how to
instantiate our results in several salient cases, in particular
continuous probabilistic transition systems and weighted transition
systems allowing negative weights.

For simplicity, we have worked exclusively with symmetric
$\V$-categories throughout; nevertheless, we stress that our results carry over
straightforwardly to the non-symmetric case, which covers
quantitative analogues of simulation preorders (indeed, some of the
existing quantitative coalgebraic Hennessy-Milner theorems already do
apply to non-symmetric distances~\cite{WS21,KKK+21,WS22}).
In fact, we expect our main expressivity theorem to be easily transported to
topological categories that admit an initial dense object (which takes the role of \(\V_s\)).
We leave this issue to future work. Another important direction is to develop a
general coalgebraic treatment of characteristic logics for
non-branching-time (e.g.\ linear-time) behavioural distances
(e.g.~\cite{AlfaroEA09,FahrenbergEA11}), possibly building on recent results in this direction~\cite{BeoharEA22}.

\clearpage
\appendix
\section{Appendix: Omitted Proofs}

\subsection{Proof of~\autoref{prop:Vs-closed}}
Since a split subobject of a separated \(\V\)-category is L-closed, we conclude
immediately that \(\Cats{\V}(X,V)\) is closed in \(\V^{|{X}|}\). Furthermore,
since the symmetrization functor \((-)_{s}\) preserves products and split
monomorphisms, the second assertion follows. \qed

\subsection{Proof of~\autoref{prop:form-sem-f}}
Just note that \(\mathcal{L}(\Lambda)\) is the initial algebra for the signature
\begin{displaymath}
  (\top,-\wedge-,-\vee-,u\otimes-,\hom_{s}(u,-),\lambda(-)),
\end{displaymath}
that \(\V_{s}^{X}\) becomes an algebra of this type with the canonical
interpretation of these symbols, and that
\(\Sema{-}\colon\mathcal{L}(\Lambda)\to \V_{s}^{X}\) is the unique algebra
homomorphism from the initial algebra to \(\V_{s}^{X}\). Since
\(-\cdot f \colon\V_{s}^{Y}\to\V_{s}^{X}\) is an algebra homomorphism, one
obtains \(\Sema{-}=\Semb{-}\cdot f\). \qed

\subsection{Proof of~\autoref{prop:bd-formula}}
  Let \(f \colon (X,a,\alpha) \to (Y,b,\beta)\) be an \(\ftF\)-coalgebra homomorphism.
  Consider the symmetric \(\V\)-category \((X,b_f^\triangleleft)\) consisting of the set \(X\) equipped with the initial structure w.r.t.\  the structured map \(f \colon X \to \ftII{(Y,b)}\).
  The following diagram of solid arrows commutes in \(\Cats{\V}_\sym\).

  \begin{displaymath}
    \begin{tikzcd}
      (X,a) & \ftF (X,a) \\
      (X,b_f^\triangleleft) & \ftF (X,b_f^\triangleleft) \\
      (Y,b) & \ftF (Y, b)
      \ar[from=1-1, to=1-2, "\alpha"]
      \ar[from=1-1, to=2-1, "1_X"']
      \ar[from=1-1, to=3-1, "f"', bend right=60]
      \ar[from=1-2, to=2-2, "\ftF 1_X"]
      \ar[from=1-2, to=3-2, "\ftF f", bend left=60]
      \ar[from=2-1, to=2-2, "\ftF 1_X \cdot \alpha"', dashed]
      \ar[from=2-1, to=3-1, "f"']
      \ar[from=2-2, to=3-2, "\ftF f"]
      \ar[from=3-1, to=3-2, "\beta"']
    \end{tikzcd}
  \end{displaymath}

  Hence, \(\ftII{\ftF f} \cdot \ftII{\ftF 1_X \cdot \alpha} = \ftII{\beta \cdot f}\).
  By definition, \(f \colon (X,b_f^\triangleleft) \to (Y,b)\) is initial, which, by hypothesis, implies that \(\ftF f \colon \ftF (X,b_f^\triangleleft) \to \ftF (Y,b)\) is initial.
  Thus, \(\ftF1_X \cdot \alpha \colon (X,b_f^\triangleleft) \to \ftF (X,b_f^\triangleleft)\) is a coalgebra homomorphism.
  Therefore, the claim follows from the fact that for all \(x,y \in X\), \(b(f(x),f(y)) = b_f^\triangleleft(x,y)\).
\qed

\subsection{Proof of~\autoref{prop:bd-stable}}
  Let \(x,y \in X\).
  Clearly, \(bd_\alpha^\ftF (x,y) \geq bd_\beta^\ftF (f(x),f(y))\), since coalgebra homomorphisms are closed under composition.
  To see the reverse inequality, let \(g \colon (X,a,\alpha) \to (Z, c, \gamma)\) be a coalgebra homomorphism.
  The category of coalgebras for \(\ftF\) is cocomplete.
  Hence, let \(i_1 \colon (Y,b,\beta) \to (P, p, \rho)\) and \(i_2 \colon (Z,c,\gamma) \to (P,p,\rho)\) be a pushout for the span consisting of \(f\) and \(g\).
  Then, \(i_1\) is a coalgebra homomorphism and \(p(i_1 \cdot f(x),i_1 \cdot f(y)) = p( i_2 \cdot g (x), i_2 \cdot g(y)) \geq c(g(x),g(y)\).
\qed

\subsection{Proof of~\autoref{p:25}}
  Let \(f \colon (X,a,\alpha) \to (Y,b,\beta)\) be a morphism of \(\ftF\)-coalgebras, and \(x,y\) be elements of~\(X\).
  Since formulas evaluate to $\V$-functors, \(b(f(x),f(y)) \leq ld_\beta^\Lambda(f(x),f(y))\).
  Therefore, the claim follows from the fact that logical distance is invariant under coalgebra homomorphisms.

\subsection{Proof of~\autoref{lem:ld-le-bd}}
By Remark~\ref{p:23}, it suffices to show that the map \(\ftF 1_X \cdot \alpha\) is a \(\V\)-functor from \((X,ld_\alpha^\Lambda)\) to \(\ftF (X,ld_\alpha^\Lambda)\).
Let \(\lambda \in \Lambda\) and \(\phi \in \calL(\Lambda)\).
The following diagram of solid arrows commutes in \(\Cats{\V}_\sym\).
\begin{center}
  \begin{tikzcd}[row sep=small, column sep=huge]
    (X,a) & (X,ld_\alpha^\Lambda) & \\
    & & \V_s. \\
    \ftF (X,a) & \ftF (X,ld_\alpha^\Lambda) & \V_s\\
    \ar[from=1-1, to=1-2, "1_X"]
    \ar[from=1-1, to=3-1, "\alpha"']
    \ar[from=1-1, to=2-3, "\Sema{\lambda(\phi)}", bend left=50]
    \ar[from=1-2, to=3-2, dashed, "\ftF 1_X \cdot \alpha"']
    \ar[from=1-2, to=2-3, "\Sema{\lambda(\phi)}", near start]
    \ar[from=3-1, to=3-2, "\ftF 1_X"']
    \ar[from=3-1, to=2-3, "\lambda(\Sema{\phi})"', bend right=50]
    \ar[from=3-2, to=2-3, "\lambda(\Sema{\phi})"', near start]
   \end{tikzcd}
\end{center}
Hence, \(\ftII{\lambda(\Sema{\phi})} \cdot \ftII{\ftF 1_X \cdot \alpha} = \ftII{\Sema{\lambda(\phi)}}\).
Therefore, the claim follows by the universal property of initial cones.
\qed

\subsection{Proof the $\cFun$ is the Greatest \(\V_s\)-closure Operator}
\begin{proposition}
  \label{p:5}
  Let \(\cC\) be a \(\V_s\)-closure operator.
  Then, for every set \(X\) and every set \(A \subseteq \V^X\), \(\cC(A) \subseteq \cFun(A)\).
\end{proposition}
\begin{proof}
	Let \(A\) be a subset of \(\V^X\).
	Then, \(A \subseteq \cFun (A)\) and \(\cFun (A) = \Cats{\V}_\sym(X_A,\V_s)\) is \(\cC\)-closed because \(\cC\) is a \(\V_s\)-closure operator.
	Therefore, \(\cC(A) \subseteq \cFun (A)\).
\end{proof}

\subsection{Proof of~\autoref{lem:closure-dense}}
\begin{enumerate}
		\item Clearly, if \(A\) is initial then \(A\) is \(\cFun\)-dense. If \(A\) is \(\cC\)-dense
		then it is initial, since~\(\V_s\) is initial dense in \(\Cats{\V}_\sym\) and
		\(\cFun\) is the \emph{greatest} $\V_s$-closure operator.

		\item By the previous clause, if \(\cC(A)\) is initial then it is \(\cFun\)-dense;
		that is \(\Cats{\V}_\sym(X,\V_s) = \cFun(\cC(A))\). Using the fact that
		\(\cFun\) is the \emph{greatest} $\V_s$-closure operator,
		\[
			\Cats{\V}_\sym(X,\V_s) = \cFun(\cC(A)) \subseteq \cFun(\cFun(A)) = \cFun(A) \subseteq \Cats{\V}_\sym(X,\V_s).
		\]
		Therefore, by the previous clause, \(A\) is initial.\qed
\end{enumerate}

\subsection{Proof of~\autoref{prop:charc-ini}}
\begin{lemma}\label{lem:pp-closed}
  Let \(A \subseteq \V_s^X\) be a propositional algebra.
  For every \(\V\)-functor \(i \colon Y \to X\), \(A \cdot i \subseteq \V_s^Y\) is a propositional algebra.
\end{lemma}
\begin{proof}
  Being left and right adjoint, the \(\V\)-functor
  \(\V^{X}\to\V^{Y},\,\varphi\mapsto\varphi\cdot i\) preserves finite suprema
  and finite infima as well as the operations \(u\otimes-\) and \(\hom(u,-)\).
  Therefore the \(\V\)-functor \(-\cdot i \colon\V_{s}^{X}\to\V_{s}^{Y}\)
  preserves the operations \(\top\), \(\wedge\), \(\vee\), \(\hom_s(u,-)\), and
  \(u \otimes -\), for every \(u \in \V\). Consequently, the image \(A\cdot i\)
  of every propositional algebra \(A \subseteq \V_s^X\) is a propositional
  algebra.
\end{proof}
  Let \(A \subseteq \V_s^X\) be a propositional algebra that is also an initial cone, and let \((i \colon X_i \to X)_{i \in \calI_X}\) be a cocone of initial \(\V\)-functors that coreflects \(\cC\)-density and \(X_i \in \calI\) for each \(i \in \calI_X\).
  Then, for every \(i \in \calI\), \(A \cdot i\) is an initial algebra, by Lemma~\ref{lem:pp-closed}, that is also an initial cone, as \(A\) and \(i\) are initial.
  Thus, for every \(i \in \calI\), \(A \cdot i\) is \(\cC\)-dense, because \(\cC\) characterizes initiality on \(\calI\).
  Therefore, since the cone \((i \colon X_i \to X)_{i \in \calI}\) coreflects \(\cC\)-density, \(A\) is \(\cC\)-dense.

\qed

\subsection{Proof of~\autoref{prop:coref-ini-preserve}}
  Let \(X\) be a symmetric \(\V\)-category in \(\calI\) and \((i \colon X_i \to X)_{i \in \calI_X}\) be a cocone of initial \(\V\)-functors such that \(X_i \in \calI\) and the cocone \((\ftF i \colon \ftF X_i \to \ftF X)_{i \in \calI_X}\) coreflects initiality.
  We will show that the cocone \((\ftG i \colon \ftG X_i \to \ftG X)_{i \in \calI_X}\) also coreflects initiality.
  To this end, let \(A \subseteq \Cats{\V}_\sym(\ftG X, \V_s)\) be a set such that for every \(i \in \calI_X\), the cone \(A \cdot \ftG i\) is initial.
  Let \(i \in \calI_X\).
  Then, since \(j_{X_i}\) is initial, the cone \(A \cdot \ftG i \cdot j_{X_i}\) is initial.
  Thus, as \(j\) is a natural transformation, the cone \(A \cdot j_X \cdot \ftF i\) is initial.
  Consequently, by hypothesis, we obtain that the cone \(A \cdot j_X\) is initial.
  Therefore, since \(j_X\) is dense, by \cite[Lemma~1.10(4)]{HT10}, the cone \(A\) is initial.
\qed

\subsection{Proof of \autoref{thm:sw-L}}

We need the following
\begin{lemma}
  \label{lem:form-codir}
  Let \(X\) be a set, and \(A\) a non-empty subset of \(\SET(X,\V)\).
  Assume that \(A\) is closed under \(\wedge\) and \(\hom_s(u,\psi)\), for every \(u \in \V\).
  Then, for every \(x \in X\), the set
  \[
    \{\hom_s(\psi(x), \psi) \mid \psi \in A \}
  \]
  is codirected.
\end{lemma}
\begin{proof}
  Let \(\phi,\psi \in A\) and \(x\in X\).
  Consider \(\gamma = \hom_s(u, \phi) \wedge \hom_s(v,\psi)\), where \(u = \phi(x)\) and \(v = \psi(x)\).
  Then, for every \(y \in X\),
\begin{align*}
    \hom_s(\gamma(x),\gamma(y)) \leq&\; \hom(\gamma(x),\gamma(y))\\
     \leq&\; \hom(k,\gamma(y))\\
        =&\; \hom_s(\phi(x),\phi(y)) \wedge \hom_s(\psi(x),\psi(y)).
\end{align*}
\end{proof}

\begin{lemma}
  \label{lem:sw-main}
  Let \((X,a)\) be a finite \(\V\)-category, \(A \subseteq \V_s^{(X,a)}\) be a propositional algebra,
  and~\(u\) be an element of \(\V\) such that for all \(v \in \V\), \(\hom(u,v)\) is way above \(v\).
  If \(A\) is initial, then for every \(f \colon X \to \V_s \in \Cats{\V}_\sym\) there is \(\psi_u \in A\) such that \(u \leq [\psi_u, f]\) and \(k \leq [f,\psi_u]\).
\end{lemma}
\begin{proof}
Let \(f \colon (X,a) \to \V_s\) be a \(\V\)-functor.
  By hypothesis, for all \(x,y \in X\), \(\hom(u, a(x,y))\) is way above \(a(x,y) = \bigwedge_{\psi \in A} \hom_s(\psi(x),\psi(y))\).
  Hence, by Lemma~\ref{lem:form-codir}, for all \(x,y \in X\) there is \(\psi_{u,x,y} \in A\) such that \(u \otimes \hom_s(\psi_{u,x,y}(x),\psi_{u,x,y}(y)) \leq a(x,y)\).
  Thus, given that \(f \colon (X,a) \to \V_s\) is a \(\V\)-functor, for all \(x,y \in X\),
  \[
    u \otimes \hom_s(\psi_{u,x,y}(x),\psi_{u,x,y}(y)) \leq a(x,y) \leq \hom_s(f(x),f(y)) \leq \hom(f(x),f(y)).
  \]
  Now, since \(A\) is a propositional algebra, the map
  \[
    \psi_u = \bigvee_{x \in X}\bigwedge_{y \in X} f(x) \otimes \hom_s(\psi_{u,x,y}(x),\psi_{u,x,y})
  \]
  belongs to \(A\).
  Furthermore, for every \(z \in X\),
  \[
    u \otimes \psi_u(z) \leq \bigvee_{x \in X} u \otimes f(x) \otimes \hom_s(\psi_{u,x,z}(x),\psi_{u,x,z}(z)) \leq f(z),
  \]
  and,
  \[
    f(z) \leq \bigwedge_{y \in X} f(z) \otimes \hom_s(\psi_{u,z,y}(z),\psi_{u,z,y}(z)) \leq
  \psi_u(z).
  \]
  Therefore, \(u \leq [\psi_u,f]\) and \(k \leq [f,\psi_u]\).
\end{proof}

Let \((X,a)\) be a finite \(\V\)-category, \(A \subseteq \V_s^{(X,a)}\) be a propositional algebra that is also an initial cone,  and~\(u\) be an element of \(\V\) such that for all \(v \in \V\), \(\hom(u,v)\) is way above \(v\).
Furthermore, suppose that \(f \colon (X,a) \to \V_s\) is a \(\V\)-functor.
By Lemma~\ref{lem:sw-main}, there is \(\psi_u \in A\) such that \(u \leq [\psi_u, f]\) and \(k \leq [f,\psi_u]\).
Note that since \(\hom(u,k)\) is way above \(k\), \(u \leq k\).
Hence, \(u \leq [\psi_u, f]_s\).
Therefore, \(u \otimes u \leq [\psi_u, f] \otimes [\psi_u, f]\).
\qed

\subsection{Proof of~\autoref{thm:sw-finite}}

In a finite quantale every element is way above itself.
Therefore, the claim follows by using \(u = k\) in Lemma~\ref{lem:sw-main}.
\qed

\subsection{Proof of~\autoref{thm:main}}
First, we need to prove the following
\begin{lemma}
  \label{p:26}
  Let \(\ftF \colon \Cats{\V}_\sym \to \Cats{\V}_\sym\) be a \(\Lambda\)-Kantorovich functor that coreflects initiality over a class \(\calI\) of symmetric \(\V\)-categories, and let \(\cC\) be a \(\V_s\)-closure operator that characterizes initiality on \(\calI\).
  Furthermore, assume that every predicate lifting of \(\Lambda\) is \(\cC\)-continuous.
  For every propositional algebra \(A \subseteq \V_s^X\), if \(A\) is initial, then \(\Lambda(A)\) is initial.
\end{lemma}
\begin{proof}
  Let \(A \subseteq \V_s^X\) be a propositional algebra that is also an initial cone.
  Since \(\ftF\) coreflects initiality over \(\calI\), there is a cocone \((i \colon X_i \to X)_{i \in \calI_X}\) of initial \(\V\)-functors such that \(X_i \in \calI\) and the cocone \((\ftF i \colon \ftF X_i \to \ftF X)\) coreflects initiality.
  Hence, to show that \(\Lambda(A)\) is initial, it suffices to show that for every \(i \in \calI_X\), the cone \(\Lambda(A) \cdot \ftF i\) is initial.

  Let \(i \in \calI_X\).
  By Lemma~\ref{lem:pp-closed}, \(A \cdot i \subseteq \V_s^{X_i}\) is a propositional algebra; moreover, it is also an initial cone because \(A\) and \(i\) are initial.
  Thus, since \(\cC\) characterizes initiality on \(\calI\), \(A \cdot i\) is \(\cC\)-dense;
  that is, \(\cC(A\cdot i) = \Cats{\V}_\sym(X_i,\V_s)\).
  Then, given that every \(\lambda \in \Lambda\) is continuous,
  \(
    \cC(\Lambda(A \cdot i)) = \cC(\Lambda(\Cats{\V}_\sym(X_i,\V_s)))
  \).
  Consequently, as \(\ftF\) is \(\Lambda\)-Kantorovich, \(\cC(\Lambda(A \cdot i))\) is initial, which entails that \(\Lambda(A \cdot i)\) is initial by Lemma~\ref{lem:closure-dense}(\ref{p:42}).
  Therefore, because every predicate lifting is a natural transformation, \(\Lambda(A) \cdot \ftF i = \Lambda(A \cdot i)\) is initial.
\end{proof}
  Let \((X,a,\alpha)\) be an \(\ftF\)-coalgebra and \(\calI\) a class of symmetric \(\V\)-categories such that~\(\ftF\) coreflects initiality in \(\calI\) and \(\cC\) characterizes initiality on \(\calI\).
  By Lemma~\ref{lem:pp-closed}, the set \(\Sema{\calL(\Lambda)} \subseteq \V_s^{(X,ld_\alpha^\Lambda)}\) is a propositional algebra, which is also an initial cone by definition of logical distance.
  Hence, by Lemma~\ref{p:26}, the cone \(\Lambda(\Sema{\calL(\Lambda)})\) is initial.
  Therefore, by Lemma~\ref{lem:ld-le-bd}, \(ld_\alpha^\Lambda \leq \bd_\alpha^\ftF\).
\qed

\subsection{Proof of~\autoref{thm:kl-discrete}}
\begin{proposition}
  If \(\Lambda\) is a separating collection of predicate liftings for a functor \(\ftF \colon \SET \to \SET\), then the Kantorovich lifting \(\ftF^\Lambda \colon \EQ \to \EQ\) preserves discrete equivalence relations.
\end{proposition}
\begin{proof}
  By definition of Kantorovich lifting and separating set of predicate liftings, the cone of all maps \(\lambda(f) \colon \ftF^\Lambda (X, 1_X) \to (2,1_2)\) with \(\lambda \in \Lambda\) and \(f \colon X \to 2 \) is initial and mono.
  Therefore, \(\ftF^\Lambda(X,1_X)\) is the discrete equivalence relation on \(\ftF X\).
\end{proof}

\begin{lemma}
  \label{p:14}
  Let \(\ftbF \colon \EQ \to \EQ\) be a lifting of a functor \(\ftF \colon \SET \to \SET\) and \((X,\alpha)\) an \(\ftF\)-coalgebra.
If \(\ftbF (X,1_X) = (\ftF X, 1_{\ftF X})\) then \(bd_\alpha^{\ftbF} \leq beq_\alpha\).
\end{lemma}
\begin{proof}
  Let \(\sim\) be an equivalence relation on \(X\) such that \(\alpha \colon (X,\sim) \to \ftbF(X,\sim) \in \EQ\).
  Consider the quotient map \(q \colon X \to X/{\sim}\) and a section \(s \colon X/{\sim} \to X\) of \(q\).
  Then, both \(q \colon (X,\sim) \to (X/{\sim}, 1_{X/{\sim}})\) and \(s \colon (X/{\sim},1_{X/{\sim}}) \to (X,1_X)\) belong to \(\EQ\).
  Hence, \(\ftbF(s\cdot q) \cdot \alpha \colon (X,\sim) \to (\ftF X, 1_{\ftF X})\) belongs to \(\EQ\).
  Consequently, \(\ftF(s \cdot q) \cdot \alpha \cdot s \cdot q = \ftF(s\cdot q) \cdot \alpha\).
  That is, the map \(\ftF(s \cdot q) \cdot \alpha\) determines an \(\ftF\)-coalgebra \((X,\beta)\) such that \(s \cdot q \colon (X,\alpha) \to (X,\beta)\) is a morphism of \(\ftF\)-coalgebras, which implies \(\sim\, \leq beq_\alpha\).
  Therefore, \(bd^{\ftbF}_\alpha \leq beq_\alpha\).
\end{proof}

\autoref{thm:kl-discrete} is now and immediate consequence of Proposition~\ref{prop:bd-stable} and Lemma~\ref{p:14}.
\qed

We also record the following observation.
\begin{theorem}
  Let \(\ftbF \colon \EQ \to \EQ\) be a lifting that preserves initial morphisms of a functor \(\ftF \colon \SET \to \SET\) that admits a terminal coalgebra \((Z,\gamma)\).
  Then, for every \(\ftF\)-coalgebra \((X,\alpha)\), \(bd_\alpha^\ftbF = beq_\alpha\) iff \(\ftbF(Z,1_Z) = (\ftF Z, 1_{\ftF Z})\).
\end{theorem}
\begin{proof}
  Note that, by definition of terminal coalgebra, \(beq_\gamma = 1_Z\).
  Hence, if \(bd_\gamma^{\ftbF} = beq_\gamma\) then \(\ftbF(Z,1_Z) = (\ftF Z, 1_{\ftF Z})\), since \(\gamma \colon (Z,bd_\gamma^\ftbF) \to \ftbF(Z,bd_\gamma^\ftbF)\) is an isomorphism in \(\EQ\).
  The converse statement follows by Proposition~\ref{prop:bd-stable} and Lemma~\ref{p:14}.
\end{proof}

\subsection{Proof of~\autoref{thm:sw-cInfsup}}
First, we show the following
\begin{lemma}
  \label{p:20}
  Suppose that \(\V\) satisfies the condition that for every \(u \in \V\) the map \(u \otimes -\) preserves codirected infima.
  Furthermore, let \(X\) be a symmetric \(\V\)-category and \(A \subseteq \V_s^X\) be an initial cone.
  For every \(f \colon X \to \V_s\),
  \[
    f = \bigvee_{x \in X} \bigwedge_{\psi \in A} \phi_{x,\psi},
  \]
  where, \(\phi_{x,\psi} \colon X \to \V_s\) is the \(\V\)-functor \(f(x) \otimes \hom_s(\psi(x), \psi)\).
\end{lemma}
\begin{proof}

Now, since \(A\) is initial, for all \(x,y \in X\),

  \[
    \bigwedge_{\psi \in A} \hom_s(\psi(x),\psi(y)) = a(x,y) \leq \hom_s(f(x),f(y)) \leq \hom(f(x),f(y)).
  \]
  Hence, by hypothesis on \(\V\), for every \(y \in Y\),
  \begin{align*}
    f(y) &= \bigvee_{x \in X} \left ( f(x) \otimes \bigwedge_{\psi \in A} \hom_s(\psi(x),\psi(y)) \right) \\
         &= \bigvee_{x \in X} \bigwedge_{\psi \in A} f(x) \otimes \hom_s(\psi(x),\psi(y)).
  \end{align*}
\end{proof}

  Let \(X\) be a symmetric and finite \(\V\)-category and \(f \colon X \to \V_s\) be a \(\V\)-functor.
  Additionally, let \(A \subseteq \V_s^X\) be a propositional algebra that is also an initial cone.
  By Lemma~\ref{p:20},
  \[
    f = \bigvee_{x \in X} \bigwedge_{\psi \in A} \phi_{x,\psi},
  \]
  where, \(\phi_{x,\psi} \colon X \to \V_s\) is the \(\V\)-functor \(f(x) \otimes \hom_s(\psi(x), \psi)\).
  Since \(A\) is a propositional algebra, for every \(x \in X\) the set \(\{\phi_{x,\psi} \mid \psi \in A\}\) is contained in \(A\) and, by Lemma~\ref{lem:form-codir}, it is codirected.
  Therefore, \(f \in \ccInf(A)\).
\qed

\subsection{Proof of~\autoref{thm:sw-cInf}}
  Let \(X\) be a symmetric and finite \(\V\)-category and \(f \colon X \to \V_s\) be a \(\V\)-functor.
  Additionally, let \(A \subseteq \V_s^X\) be a propositional algebra that is also an initial cone.
  By Lemma~\ref{p:20},
  \[
  	f = \bigvee_{x \in X} \bigwedge_{\psi \in A} \phi_{x,\psi},
  \]
  where \(\phi_{x,\psi} \colon X \to \V_s\) is the \(\V\)-functor \(f(x) \otimes \hom_s(\psi(x),\psi)\).
  Hence, since \(\V\) is completely distributive,
   \[
  	f = \bigwedge_{g \in C} \bigvee_{x \in X} \phi_{x, g(x)},
  \]
  where, \(C\) is the set of choice functions that for every \(x \in X\) choose an index \(g(x) \in A\).
  Therefore, as \(X\) is finite and \(A\) is a propositional algebra, \(f \in \cInf(A)\).
\qed

\end{document}